\documentclass[12pt]{article}
\usepackage{euscript,amsfonts,amsmath,amssymb,color}
\usepackage{graphicx}
\usepackage{natbib}
\usepackage{caption}
\usepackage[title]{appendix}

\usepackage{fullpage}
\usepackage{soul}
\usepackage{bbm, dsfont}
\usepackage{comment}

\definecolor{webgreen}{rgb}{0,0.4,0}
\definecolor{webbrown}{rgb}{0.6,0,0}
\definecolor{purple}{rgb}{0.5,0,0.25}
\definecolor{darkblue}{rgb}{0,0,0.7}
\definecolor{darkred}{rgb}{0.7,0,0}
\definecolor{darkgreen}{rgb}{0,0.7,0}
\usepackage[pdfborder=false]{hyperref}
\hypersetup{colorlinks,citecolor=darkred,filecolor=black,linkcolor=darkblue,urlcolor=webgreen,pdfpagemode=none,pdfstartview=FitH}
\newcommand{\ignore}[1]{}
\newtheorem{lemma}{{\bf Lemma}}
\newtheorem{proposition}{{\bf Proposition}}
\newtheorem{corollary}{{\bf Corollary}}
\newtheorem{theorem}{{\bf Theorem}}
\newtheorem{defn}{{\bf Definition}}
\newtheorem{obs}{{\bf Observation}}

\newenvironment{proof}{\noindent \textsl{Proof.\/}}
{\hfill $\blacksquare{}$ \vspace{12pt}}

\title{{\bf Undominated monopoly regulation}\thanks{We are grateful to the editor and two anonymous referees for their helpful suggestions. We thank Dirk Bergemann, Sushil Bikhchandani, Sulagna Dasgupta, Rahul Deb,  Yingni Guo, Mallesh Pai, Alessandro Pavan, Jeevant Rampal, Bruno Strulovici, and  Vilok Taori for helpful discussions. We also thank seminar participants at Academia Sinica, Ashoka University, Higher School of Economics at St. Petersburg, Shiv Nadar University, University of Liverpool, Conference on Mechanisms and Institution Design in Budapest, and Meeting of the Society for Social Choice and Welfare in Paris for their comments.}}
\author{Debasis Mishra and Sanket Patil\footnote{Mishra: Indian Statistical Institute, Delhi (\href{mailto:dmishra@gmail.com}{dmishra@gmail.com}). Patil: Indian Institute of Management Bangalore (\href{mailto:sanket.patil@iimb.ac.in}{sanket.patil@iimb.ac.in}).}}
\begin{document}
\allowdisplaybreaks
\maketitle

\begin{abstract}
We study undominated mechanisms with transfers for regulating a monopolist who privately observes the marginal cost of production. A mechanism is undominated if no other mechanism gives the regulator a strictly higher payoff at some marginal cost of the monopolist without lowering the regulator’s payoff at other costs. We show that an undominated mechanism has a quantity floor: whenever the monopolist is allowed to operate, it produces above a threshold quantity. Moreover, the regulator's operation decision is stochastic only if the monopolist produces at the quantity floor. We provide a near-complete characterization of the set of undominated mechanisms and use it to (a) derive a max-min optimal regulatory mechanism, (b) provide a foundation for deterministic mechanisms, and (c) show that the efficient mechanism is dominated.
\end{abstract}



\newpage

\section{Introduction}
\label{sec:intro}

We consider the problem of regulating a monopolist who privately observes the marginal cost of production. The regulator’s payoff is a weighted sum of the consumer surplus and the monopolist’s profit, with a strictly smaller weight on the monopolist’s profit. Given such preferences, the regulator would like to maximize the consumer surplus irrespective of the cost (type) of the monopolist. However, incentive constraints imply that increasing consumer surplus at a given type reduces the consumer surplus at lower types. A seminal paper, \citet{BM82}, studies the optimal resolution of this trade-off for a regulator who uses transfers and maximizes expected payoff with respect to a prior over the plausible marginal costs. Most of the subsequent literature, both with and without transfers, assumed expected payoff maximization by the regulator subject to incentive (IC) and participation (IR) constraints of the monopolist (see \citet{AS06} for an extensive survey).

This paper takes the first step in understanding the prior-free criterion of undomination in the context of regulatory mechanisms. A regulatory mechanism must specify an {\sl operation decision}, which may be stochastic; a {\sl quantity decision}, which is contingent on the monopolist operating; and a {\sl transfer decision} on the amount of subsidy/tax to the monopolist, which is also contingent on the monopolist operating. A regulatory mechanism is undominated if no other mechanism gives the regulator a strictly higher payoff at some marginal cost of the monopolist without lowering the regulator’s payoff at other costs. Such prior-free undominated mechanisms are relevant particularly when the regulator lacks confidence in the belief over the plausible marginal costs or fears ex-post criticism of the subjective belief used while choosing the regulatory mechanism. The concept of undomination is also useful for refining the set of worst-case optimal mechanisms when it is not a singleton.

Our first main result shows that an undominated regulatory mechanism must be a {\sl floor-randomized mechanism}. A floor-randomized mechanism is defined by partitioning the type space into three disjoint intervals (with the possibility that some of these intervals may be empty) and a quantity floor $\hat{q}$, which does not depend on the marginal cost of the monopolist. In the left interval (containing the low-cost types), the monopolist (firm) operates for sure and produces quantity higher than the quantity floor $\hat{q}$. In the middle interval, the firm operates with positive probability (possibly less than one) and produces quantity $\hat{q}$. The firm is shut in the right interval, which contains high-cost types. Thus, the operation decision is stochastic only when the monopolist produces a quantity equal to the quantity floor $\hat{q}$.

A general incentive compatible (IC) and individually rational (IR) mechanism is characterized by the monotonicity of the product of quantity and operation probability, and an envelope formula specifying the transfer. On the other hand, floor-randomized mechanisms (which are IC and IR) have several structural properties that make them easier to handle: (i) whenever the monopolist is allowed to operate, the quantity produced is bounded below by $\hat{q}$; (ii) operation decision is stochastic only when the quantity produced equals $\hat{q}$; (iii) both quantity and operation decisions are monotone. 

Using these properties of floor-randomized mechanisms, we derive two further necessary conditions for undomination. The first condition is what we call {\sl downward distortion}. It requires that the quantity produced at every type be lower than the efficient quantity at that type. Thus, there is a type-dependent upper bound on the quantity to be produced and a type-independent lower bound $\hat{q}$. The second condition is technical and requires the product of the quantity and the operation decisions to be left continuous. Thus, we show that every undominated mechanism is a floor-randomized mechanism satisfying downward distortion and left-continuity. Conversely, every floor-randomized mechanism satisfying left-continuity and a stricter version of downward distortion is undominated.\footnote{The stricter version of downward distortion requires that the quantity produced at every type must be {\sl strictly} lower than the efficient quantity at that type.} Thus, undomination, despite being a weak criterion, puts considerable structure on the set of IC and IR mechanisms.

The weight put on the monopolist's payoff by the regulator captures the intensity of the regulator's distributional preferences. A larger weight corresponds to a weaker distributional preference (or stronger preference for efficiency). We show that the set of undominated mechanisms expands as the regulator's distributional preferences weaken. However, our necessary and sufficient conditions are invariant in this preference parameter. This highlights that the set of undominated mechanisms can be divided into two subsets. One that satisfies our sufficient conditions and is unaffected by the regulator's distributional preferences, and the other does not satisfy our sufficient conditions but expands as the regulator puts more weight on the monopolist's profit. Surprisingly, the efficient mechanism, where the monopolist always operates and produces efficient quantity (maximizing total surplus), belongs to neither of these subsets as it is dominated.

Our analysis has several implications for regulation design. The first is for an ambiguity-averse regulator who evaluates any mechanism by the worst-case expected payoff it generates over a set of priors on the monopolist's marginal cost. A {\sl max-min optimal} mechanism has the highest worst-case payoff. We establish that if the regulator's set of priors contains a prior that (first-order) stochastically dominates every other prior in the set, then the max-min optimal mechanism is the optimal mechanism for that prior. 

As an essential step in obtaining the max-min optimal regulatory mechanism, we show that undominated mechanisms are {\sl monotonic}. The regulator's expected payoff increases as the distribution over marginal cost decreases in the first-order stochastic dominance sense. In general, we show that an arbitrary IC regulatory mechanism need not be monotonic. This observation contrasts sharply with IC mechanisms for selling a single good. As \cite{HR15} shows, every IC  mechanism for selling a single good is monotonic.

The second implication of our analysis is that every undominated mechanism satisfying a mild continuity property maximizes the expected regulator surplus for some prior of the regulator. As a result, using an undominated mechanism can be justified as if the regulator had a prior over the marginal cost of the monopolist and was maximizing the expected surplus with respect to this prior as in \citet{BM82}.

The third implication of our results is for an optimal regulatory mechanism that maximizes the expected payoff of the regulator for a given prior. As it is without loss to search for an optimal mechanism in the set of undominated mechanisms, an optimal mechanism exists that satisfies downward distortion and has a quantity floor. Thus, our analysis identifies prior-invariant properties of optimal mechanisms. The optimal mechanism identified in \cite{BM82} satisfies these properties. However, they do not explicitly refer to the quantity floor, which is hidden in the shut-down condition of their optimal mechanism.

Finally, the set of undominated mechanisms contains stochastic mechanisms. However, we show that for every IC and IR mechanism $M$ and for every prior over the marginal costs, there exists a {\sl deterministic} IC and IR mechanism that generates more expected payoff for the regulator than $M$. This provides a foundation for deterministic mechanisms in our model and explains why the optimal regulatory mechanism in \citet{BM82} is deterministic.

Section \ref{sec:lit} discusses related literature in detail. However, to put our contribution in perspective, we now contrast it with \cite{GS24}, an important contribution to prior-free regulation design. They analyze robust regulatory mechanisms in a setup where the regulator does not know the inverse demand function or the monopolist's cost function. In their paper, the regulator evaluates mechanisms by their worst-case regret, which is the difference between the complete information and the actual payoffs of the regulator. They show that a price-cap mechanism is robustly optimal even in a setting where transfers are feasible. Our approach complements their analysis as it is also non-Bayesian but differs in the robustness criterion. Undomination, despite being a very weak robustness criterion, puts significant restrictions like floor randomization and downward distortion on the set of IC and IR mechanisms. Further, the monotonicity properties of undominated mechanisms has implications for the design of max-min optimal mechanisms. Recent literature on max-min optimal mechanism design is extensively reviewed in \citet{C19} -- also see \citet{C17} who find a max-min optimal mechanism in the multidimensional screening model.\footnote{Other recent approaches to robustness include \citet{BHM23}, where a buyer evaluates a procurement mechanism based on the ratio of actual value from the mechanism to its value in the complete information benchmark. In \citet{MPP25}, the buyer, unsure of his model, employs a lexicographic approach: first, he creates a shortlist of all max-min optimal mechanisms, and then, he uses his model to select a mechanism from that shortlist.} 

The rest of the paper is organized as follows. Section \ref{sec:model} describes the model. In Section \ref{sec:floor_rand}, we discuss floor-randomized mechanisms and show it is without loss to restrict our search of undominated mechanisms in the class of floor-randomized mechanisms. Section \ref{sec:characterization} gives a near-complete characterization of undominated mechanisms. We discuss the implications of our results for  regulation design in Section \ref{sec:implications}. Finally, related literature is discussed in Section \ref{sec:lit}, and Appendices \ref{app:floor_rand}-\ref{app:implications} contain all the proofs, and Appendix \ref{sec:fixed} provides an extension.

\section{Model}
\label{sec:model}
A regulator oversees a market with a unit mass of consumers and a monopolist selling a single product. We begin by describing the market and then discuss the regulatory mechanisms. 

\subsection{Market}

The total cost incurred by the monopolist in producing a quantity $q >0$ of the product is given by $(c + \theta q)$. The parameter $c > 0$ denotes the fixed cost which is incurred only when $q > 0$. The parameter $\theta > 0$ denotes the marginal cost of production for the monopolist. The monopolist is perfectly informed about its fixed cost and marginal cost.

Let $P:\Re_{+} \rightarrow \Re_{+}$ be a strictly decreasing inverse demand function that represents the preferences of the unit mass of consumers. If the monopolist charges a uniform price of $P(q)$, where $q>0$, exactly $q$ units of the product are sold in the market, and the consumers get a total value of
\begin{align*}
V(q) &= \int \limits_0^q P(z)dz.
\end{align*}
The associated consumer surplus (CS) is given by $\big(V(q) - q P(q)\big)$ and the profit  of the monopolist is $\big(q P(q) -  c-\theta q \big)$. Finally, the total surplus (TS) is the sum of consumer surplus and profit of the monopolist, i.e., 
\begin{align*} 
\textsc{TS}(\theta,q)= V(q) - c-\theta q.
\end{align*}
If zero units are sold, i.e. $q=0$, the consumer surplus, monopolist's profit, and total surplus are all zero.

Throughout the paper, we impose the following assumptions on the primitives. We assume that the marginal cost $\theta$ can take values in the interval $\Theta \equiv [\underline{\theta},\overline{\theta}]$, where $0 < \underline{\theta} < \overline{\theta}$. Further, $P$ satisfies the following conditions.
\begin{enumerate} 
\item[A1.] $P$ is continuous and $P(\overline{q})=0$ for some $\overline{q}>0$.
\item[A2.] $c$ and $\overline{\theta}$ are such that $\textsc{TS}\big(\overline{\theta},P^{-1}(\overline{\theta})\big) > 0$.
\end{enumerate}
Assumptions [A1] is a technical assumption on $P$. Assumption [A2] is a joint condition on $c$, $\overline{\theta}$ and $P^{-1}$. Intuitively, it says that it is socially optimal to let the most inefficient type operate in the complete information benchmark.

\subsection{Regulatory mechanisms}

The regulator perfectly knows the fixed cost of the monopolist and the inverse demand function. However, when it comes to the marginal cost, the regulator only knows that it belongs to $\Theta$. 

The regulatory apparatus employed in disciplining the behavior of the monopolist can be summarized through a regulatory mechanism $M$. 
By revelation principle, we focus on direct regulatory mechanisms in which the monopolist first reports the marginal cost $\theta$ to the regulator, who then enforces two decisions $r(\theta)$ and $q(\theta)$ on the monopolist and transfers a contingent subsidy $s(\theta)$ to the monopolist. More specifically, a regulatory mechanism is given by $M \equiv (r, q, s)$, where

\begin{itemize}

\item {\sl Operation decision.} The map $r: \Theta \rightarrow [0,1]$ denotes the probability with which the monopolist is allowed to operate.

\item {\sl Quantity rule.} If the monopolist is allowed to operate, the map $q: \Theta \rightarrow [0,\bar{q}]$ prescribes the quantity to be produced by the monopolist.

\item {\sl Subsidy rule.} If the monopolist is allowed to operate, the map $s: \Theta \rightarrow \Re$ represents the subsidy each type of the monopolist receives.

\end{itemize}
Throughout the paper, we focus attention on regulatory mechanisms such that $q(\theta)=0$ if and only if $r(\theta)=0$. This restriction is without loss and will simplify the exposition.\footnote{If $r(\theta)=0$, the monopolist is not allowed to operate indicating that value of $q(\theta)$ is inconsequential in determining the regulator's surplus which will be given by equation (\ref{eq:RS}). Similarly, if $q(\theta)=0$, the total cost of production as well as value to consumers are both $0$. Thus, for all values of $r(\theta)$ the regulator gets the same payoff.} Given such a regulatory mechanism, $M \equiv (r,q,s)$, the payoff of a type $\theta$ monopolist from truthfully reporting its type is 
\begin{align} \label{eq:profit}
u(\theta) := \Big[ q(\theta) P(q(\theta)) - c-\theta q(\theta) + s(\theta) \Big] r(\theta).
\end{align}
On the other hand, the corresponding {\sl consumer surplus (CS)} is
\begin{align}\label{eq:cs}
\textsc{CS}(\theta,M) &= \Big( V(q(\theta)) - q(\theta)P(q(\theta)) - s(\theta) \Big)r(\theta) \nonumber \\
&= \Big( V(q(\theta)) - c - \theta q(\theta)\Big)r(\theta) - u(\theta) 
\end{align}
where (\ref{eq:cs}) is obtained by substituting $s(\theta)$ from equation (\ref{eq:profit}). Observe, equation (\ref{eq:profit}) uniquely defines the subsidy $s$ for given $(r,q,u)$. Thus, it is without loss of generality to write a mechanism as $M \equiv (r,q,u)$. On the other hand, we denote a \textbf{deterministic} mechanism, i.e., a mechanism where $r(\theta) \in \{0,1\}$ for all $\theta$, as $M\equiv (q,u)$ by suppressing $r$. We follow this convention throughout the paper.

A regulatory mechanism $(r,q,u)$ is {\bf incentive compatible (IC)} if for all $\theta,\theta' \in \Theta$
\begin{align*} 
u(\theta) \ge \Big[ q(\theta') P(q(\theta')) - c - \theta q(\theta') + s(\theta') \Big] r(\theta') = u(\theta') + (\theta' - \theta) q(\theta')r(\theta').
\end{align*}
Furthermore, the regulatory mechanism $M$ is {\bf individually rational (IR)} if $u(\theta) \ge 0$ for all $\theta \in \Theta$. We now state a lemma characterizing incentive compatibility and individual rationality of the mechanism $M \equiv (r,q,u)$ in terms of its properties.

\begin{lemma}[\citet{BM82}]
\label{lem:myerson}
A mechanism $(r,q,u)$ is IC if and only if  $q(\theta)r(\theta)$ is decreasing and for all $\theta \in \Theta$ 
\begin{align*} 
u(\theta) &= u(\bar{\theta}) + \int \limits_{\theta}^{\bar{\theta}} q(z)r(z) dz.
\end{align*}
Further, $(r,q,u)$ is IR if and only if $u(\overline{\theta}) \ge 0$.
\end{lemma}
In standard screening problems, a single monotone function characterizes incentive compatibility (for instance, in single-object auctions, it is the increasing allocation rule). By contrast, in our setup, the incentive compatibility is characterized by the product $q(\theta)r(\theta)$ of two potentially non-monotonic functions $q$ and $r$. By Lemma \ref{lem:myerson}, this product must decrease in $\theta$ for a mechanism $(q,r,u)$ to be IC.

\subsection{Regulator's surplus and undominated mechanisms}
If the monopolist is of type $\theta$, the regulator's surplus from an IC and IR mechanism $M \equiv (r,q,u)$ is
\begin{align*} 
\textsc{RS}_{\alpha}(\theta,M) &= \textsc{CS}(\theta,M) + \alpha u(\theta)
\end{align*}
where $\alpha \in [0,1)$ is the weight the regulator puts on the monopolist's payoff. We can further simplify the regulator's surplus using equation (\ref{eq:cs}) to 
\begin{align} 
\textsc{RS}_{\alpha}(\theta,M) &= r(\theta)\Big( V(q(\theta)) - c - \theta q(\theta)\Big) - (1-\alpha) u(\theta) \label{eq:RS}\\
&= r(\theta)\textsc{TS}(\theta,q(\theta)) - (1-\alpha) u(\theta). \label{eq:RS2}
\end{align}
Thus, the regulator's surplus  can be separated into two parts: the expected total surplus, which is independent of $\alpha$, and the weighted disutility from the rent paid to the monopolist. The regulator is concerned about expected total surplus as well as how the expected total surplus is distributed across the consumers and the monopolist. As $\alpha$ decreases, the distributional concerns of the regulator get stronger. The extreme scenario is when $\alpha=0$, where the regulator's surplus equals consumer surplus.\footnote{If $\alpha =1$, the regulator's surplus reduces to the total surplus. In this case, the regulator is concerned only about efficiency, and it is well known that the efficient quantity is produced in the optimal mechanism.}


We now introduce a prior-free notion of comparing two IC and IR mechanisms.
\begin{defn} 
An IC and IR mechanism $\widetilde{M}$ {\bf dominates} an IC and IR mechanism $M$ if 
\begin{align*}
\textsc{RS}_{\alpha}(\theta,\widetilde{M}) \ge \textsc{RS}_{\alpha}(\theta,M)~\qquad~\forall~\theta \in \Theta
\end{align*}
with strict inequality holding for some $\theta$. In this case, we say $M$ is {\bf dominated}.

A mechanism $M$ is {\bf undominated} if it is IC and IR and not dominated.
\end{defn}

In what follows, we aim to characterize the set of undominated mechanisms.

\subsection{Remarks on the model}
Some remarks are in order about the modelling of regulatory mechanisms. First, it appears that the regulatory mechanism only controls quantity, and price is determined from inverse demand function. However, we can recast the mechanism to control price, and quantity can be determined by the demand function (as is done in \citet{AS04}).
    
Second, our regulatory mechanisms use transfers. These transfers come from the consumers and go to the monopolist (if they are negative, these are taxes collected from the monopolist that go to the consumers). Implicitly, we assume {\sl a budget-balance constraint} for the regulator. This is standard in the literature that follows \citet{BM82} - for departures, see \citet{LT86}. 

Our results extend to cost functions of the form $k_0 + (c_0+c_1 \theta)q$, where $c_1 > 0$. However, for simplicity of exposition we set $c_0 = 0$, $c_1 = 1$, and $k_0=c$ for the rest of the paper. Similar to \citet{AB22}, this simplification allows us to transparently convey the intuition for our results.\footnote{\citet{BM82} assumes a more general cost function, $k_0 + k_1 \theta + (c_0+c_1 \theta)q$, where private information affects both the marginal and fixed costs. Intuitively, in this case, it is not without loss of generality to restrict the search for the undominated mechanisms to a smaller, well-behaved subset of IC and IR mechanisms, further complicating the analysis. In Appendix \ref{sec:fixed}, we discuss these challenges in detail and exhibit they can be overcome in a polar case of the bilinear cost function $\theta + cq$, where the fixed cost $\theta$ is the private information of the monopolist, and $c > 0$ is the commonly known marginal cost.}

\section{Floor-randomized mechanisms}\label{sec:floor_rand}

This section introduces a subclass of IC and IR mechanisms called {\sl floor-randomized mechanisms}.
Every IC and IR mechanism can be transformed into a unique floor-randomized mechanism. Our main result will show that converting an arbitrary mechanism to a floor-randomized mechanism improves the regulator's surplus at every type.

As an essential step in defining a floor-randomized mechanism, we first describe the quantity floor, $\hat{q}$, which is the unique solution to the following equation:
\begin{align}\label{eq:hatq}
    V(q) - q P(q) &= c.
\end{align}
The LHS is the consumer surplus when the monopolist produces quantity $q$ in absence of any regulation. The equation solves for quantity $\hat{q}$ where this consumer surplus equals fixed cost. To see that a unique positive solution exists, note that $V(q)-qP(q)$ is a strictly increasing function with $\lim_{q \rightarrow 0} V(q) - qP(q)$ equals zero and $V(\overline{q})-\overline{q} P(\overline{q}) > c$, where the latter inequality follows from Assumption [A2].

Next, we define a floor-randomized mechanism where we use the notation $I \succ I'$ for any pair of disjoint intervals $I,I' \subseteq \Theta$ to mean $x > y$ for every $x \in I$ and every $y \in I'$.

\begin{defn}\label{def:fr}
A mechanism $(r,q,u)$ is {\bf floor randomized} if it is IC, IR, $u(\bar{\theta})=0$, and the type space $\Theta$ can be partitioned into three disjoint intervals $\Theta_1,\Theta_{01},\Theta_0$ such that $\Theta_0 \succ \Theta_{01} \succ \Theta_1$ and
\begin{align*}
    q(\theta) &\ge \hat{q}~~{\rm and}~~r(\theta)=1~\qquad~\forall~\theta \in \Theta_1 \\
    q(\theta) &= \hat{q}~~{\rm and}~~r(\theta) \in (0,1)~\qquad~\forall~\theta \in \Theta_{01} \\
    q(\theta) &= 0~~{\rm and}~~r(\theta)=0~\qquad~\forall~\theta \in \Theta_0.
\end{align*}
\end{defn}

Figure \ref{fig:floor_randomized} plots $q(\theta)$ (solid line) and $r(\theta)$ (dash-dotted line) on the vertical axes and illustrates the interval partitioning of the type space in a floor-randomized mechanism. It also sheds light on several other features of a floor-randomized mechanism: 
\begin{itemize}
\item {\sl quantity floor}: whenever the monopolist is allowed to operate with a positive probability, it produces a quantity weakly above $\hat{q}$. Thus, if $\Theta_{01}$ is empty, a deterministic mechanism with quantity floor $\hat{q}$ is also floor randomized.

\item {\sl floor randomization}: if the operation decision is random for some type of the monopolist, this type is prescribed quantity equals the floor $\hat{q}$.

\item {\sl individual monotonicity}: in a floor-randomized mechanism both $q$ and $r$ are decreasing.
\end{itemize}

Starting with an arbitrary IC and IR mechanisms $M \equiv (r,q,u)$, which need not be floor randomized, we can obtain a floor-randomized mechanism as follows: Let $\tilde{q}(\theta)\equiv q(\theta)r(\theta)$, which decreases in $\theta$ because $M$ is IC, and define
\begin{align*}
    q^T(\theta) = \begin{cases}\max(\tilde{q}(\theta),\hat{q}) & \textrm{$\tilde{q}(\theta)>0$} \\
        0 & {\rm otherwise}
    \end{cases} & &  r^T(\theta) = \begin{cases}\min(\frac{\tilde{q}(\theta)}{\hat{q}},1)& \textrm{$\tilde{q}(\theta)>0$} \\
    0 & {\rm otherwise}
    \end{cases}
\end{align*}
At every $\theta$, we have $q^T(\theta)r^T(\theta)$ equals $\tilde{q}(\theta)$. Consequently, the mechanism $M^T\equiv(r^T,q^T,u^T)$, where $u^T(\theta):= \int_{\theta}^{\overline{\theta}}\tilde{q}(z)dz$ for all $\theta$, is incentive compatible and individually rational with $u^T(\overline{\theta})=0$. Moreover, it is also floor randomized with $\Theta_{1} = \{\theta: \tilde{q}(\theta) \ge \hat{q}\}$ and $\Theta_{01} = \{\theta: \hat{q} > \tilde{q}(\theta) > 0\}$. Notice that $r^T(\theta) \in (0,1)$ is true only when $q^T(\theta)=\hat{q}$. Our first main result shows that the transformed mechanism $M^T$ improves the regulator's surplus at every type. Formally, 

\begin{figure}
    \centering
    \includegraphics[width=3in]{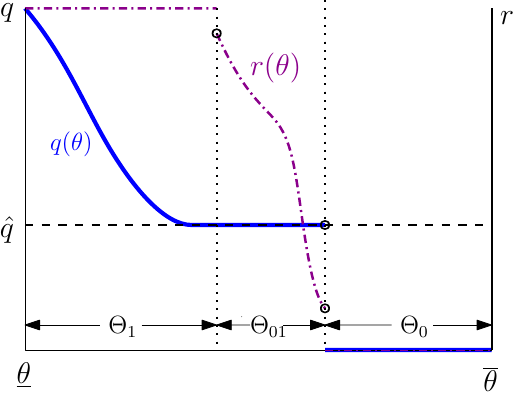}
    \caption{Graphical illustration of a floor-randomized mechanism}
    \label{fig:floor_randomized}
\end{figure}

\begin{theorem}\label{theo:fr}
Any IC and IR mechanism that is not floor randomized is dominated by a floor-randomized mechanism.
\end{theorem}

Thus, the set of undominated mechanisms is a subset of the set of floor-randomized mechanisms. The proof of Theorem \ref{theo:fr} is in Appendix \ref{app:floor_rand}. To see the main idea behind this proof, consider a deterministic IC and IR mechanisms $M$ in which every type is allowed to operate and produces a strictly positive quantity $q^\dag <\hat{q}$. Furthermore, assume that $u(\overline{\theta})=0$. Clearly, $M$ is not floor randomized. In its floor-randomized transformation, $M^T$, every type is allowed to operate with probability $q^\dag/\hat{q}$ and produces a quantity $\hat{q}$. We now argue that moving from deterministic operation to stochastic operation can make the regulator better off at every $\theta$. To see why, consider the difference
\begin{align*}
\textsc{RS}_{\alpha}(\theta, M^T)-\textsc{RS}_{\alpha}(\theta, M) &= \frac{q^\dag}{\hat{q}}\Big(V(\hat{q})-c\Big)- \Big(V(q^\dag)-c\Big)\\
&= \frac{q^\dag}{\hat{q}}\widetilde{V}(\hat{q}) + \left(1-\frac{q^\dag}{\hat{q}}\right)\widetilde{V}(0) - \widetilde{V}(q^\dag)
\end{align*}
where $\widetilde{V}(q) = \Big(V(q)-c\Big)\mathbf{1}_{\{q>0\}}$. If $\widetilde{V}(q)$ was a concave function, then Jensen's inequality would have implied that the difference is negative. However, $\widetilde{V}(q)$ is not concave because of the discontinuity at $0$. Figure \ref{fig:concav_0} plots  $\widetilde{V}(q)$, and pictorially depicts both the terms in the difference. As can be seen, the expected value of $\widetilde{V}(q)$ is denoted by the ($\blacksquare$). On the other hand, the value of $\widetilde{V}(q^\dag)$ is denoted by ($\times$), and lies at a lower level. Thus, randomization is valuable to the regulator for quantities $q < \hat{q}$ as it reduces the occurrence of the fixed cost and simultaneously increases the value to the consumers relative to a deterministic mechanism. 

\begin{figure}
    \centering
    \includegraphics[width=4.5in]{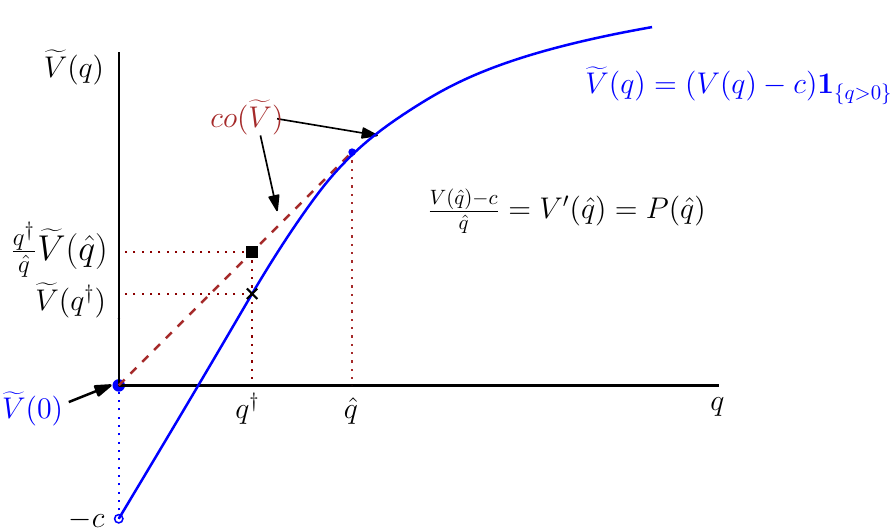}
    \caption{Graphical illustration of $\widetilde{V}$}
    \label{fig:concav_0}
\end{figure}

\section{Undominated mechanisms}\label{sec:characterization}

In this section, we provide a near-complete characterization of undominated mechanisms, and show that not every floor-randomized mechanism is undominated. We begin by introducing the {\bf efficient quantity rule}, $q_e$, which assigns every marginal cost $\theta$ the quantity that maximizes $\textsc{TS}(\theta,q)$. Thus, 
\begin{align*} 
P(q_e(\theta)) = \theta~~~\textrm{or}~~~q_e (\theta) = P^{-1}(\theta)~\qquad~\forall~\theta.
\end{align*}
Because $P$ is strictly decreasing $q_e$ is single valued, positive, 
and strictly decreasing.\footnote{At a given $\theta$, $q_e(\theta)$ is the unique maximizer of the total surplus $\textsc{TS}(\theta,q) = V(q)-c-\theta q$ because it is the only quantity satisfying the first-order condition and $\textsc{TS}(\theta,q_e(\theta)) > 0$. To see why the latter inequality holds, observe 
\begin{align*}
V(q_e(\theta)) - \theta q_e(\theta)  \ge V(q_e(\overline{\theta})) - \theta q_e(\overline{\theta})  \ge V(q_e(\overline{\theta})) - \overline{\theta} q_e(\overline{\theta})  > c, 
\end{align*} 
where the first inequality follows from $q_e(\theta)$ being strictly decreasing and the last one holds because of assumption [A2]. 
} Moreover, continuity and monotonicity of $P$ (assumption [A1]) ensures that $q_e$ is continuous.
The next observation compares $q_e(\theta)$ to the quantity floor $\hat{q}$.
\begin{obs}\label{ob:hatq_bounds}
For every $\theta$, we have $q_e(\theta)> \hat{q}~~{\rm and}~~\textsc{TS}(\theta, \hat{q})>0$.
\end{obs}
This observation, which is proven in Appendix \ref{app:characterization}, states that the quantity floor is strictly below the efficient quantity for every type. We further partition the set of floor-randomized mechanisms based on whether or not the quantity produced is weakly below the efficient quantity at all types.

\begin{defn} 
A mechanism $(r,q,u)$ satisfies {\bf downward distortion (DD)} if 
\begin{align*} 
 q(\theta) \le q_e(\theta)~\qquad~\forall~\theta 
\end{align*}
with an equality at $\underline{\theta}$. If the inequality is strict for every $\theta > \underline{\theta}$, then we say that mechanism $(r,q,u)$ satisfies {\bf strict downward distortion (strict DD)}.
\end{defn}

As we establish later in this section, the quantity rule of every undominated mechanism will satisfy DD. In addition, undominated mechanisms must also be left continuous. 

\begin{defn}
    A mechanism $(r,q,u)$ is {\bf left continuous} if the product $qr$ is left continuous.
\end{defn}

Our next main results shows that undominated mechanisms satisfy downward distortion and left-continuity.

\begin{theorem} 
\label{theo:undnec}
Suppose $M \equiv (r,q,u)$ is an IC and IR mechanism. Then, the following statements are true:
\begin{enumerate}
\item If $M$ is undominated, then it is floor randomized, left continuous, and satisfies DD.

\item If $M$ is dominated, then it is dominated by a floor-randomized and left-continuous mechanism satisfying DD.

\end{enumerate}
\end{theorem}

\begin{figure}[ht]
    \centering
    \includegraphics[width=3.8in]{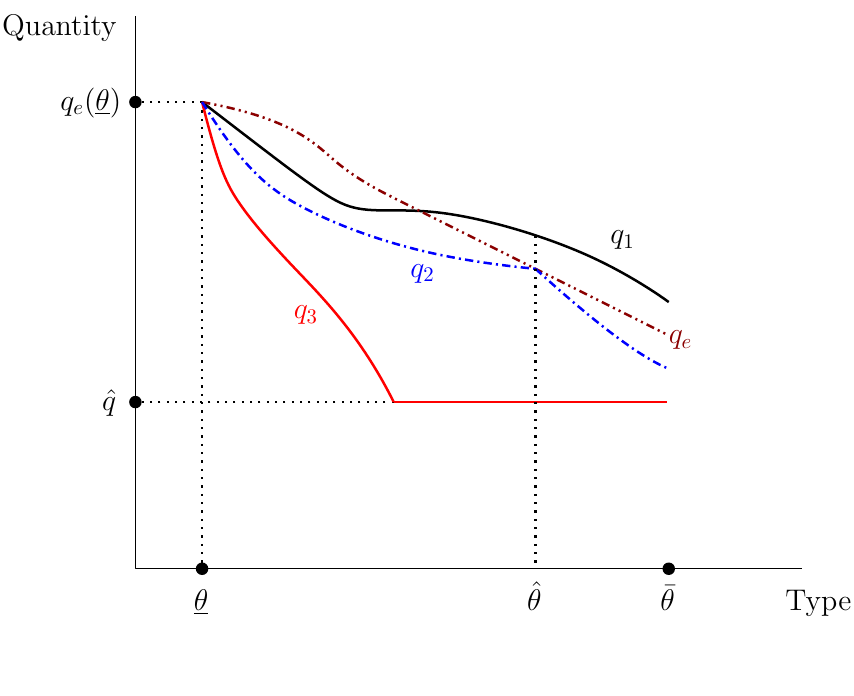}
    \caption{Illustration of necessary and sufficient conditions}
    \label{fig:Theorem2}
\end{figure}

Figure \ref{fig:Theorem2} plots quantity rules $q_1, q_2$, and $q_3$ for three floor-randomized mechanisms in which the monopolist always operates. The mechanism $M_1$ corresponding to the quantity rule $q_1$ violates DD, and by part (1) of Theorem \ref{theo:undnec}, it is dominated. To see the intuition, consider another floor-randomized mechanism $M \equiv (q, u)$ where $q(\theta) = \min \{q_1(\theta),q_e(\theta)\}$. For $\hat{\theta}$, which is shown in the figure, the total surplus is strictly larger under $M$ because $q(\theta) =q_e(\theta) < q_1(\theta)$. On the other hand, the rent paid to the monopolist is smaller because $q(\theta) \le q_1(\theta) $ for all $\theta$. The proof in the Appendix \ref{app:characterization} generalizes this intuition beyond deterministic floor-randomized mechanisms. Moreover, it  shows why left-continuity is also necessary.

Next, we show that strict DD, a stronger version of DD, and the remaining necessary conditions imply undomination. 

\begin{theorem}\label{theo:undsuff}
A floor-randomized, left-continuous mechanism satisfying strict DD is undominated.
\end{theorem}

In Figure \ref{fig:Theorem2}, the quantity rule $q_2$ satisfies DD but violates strict DD. Hence, Theorem \ref{theo:undsuff} is silent whether the mechanism corresponding to $q_2$ is dominated or not. On the other hand the quantity rule $q_3$ satisfies strict DD and is left continuous. Hence, the corresponding floor-randomized mechanism is undominated.

\begin{figure}
    \centering
    \includegraphics[width=4.5in]{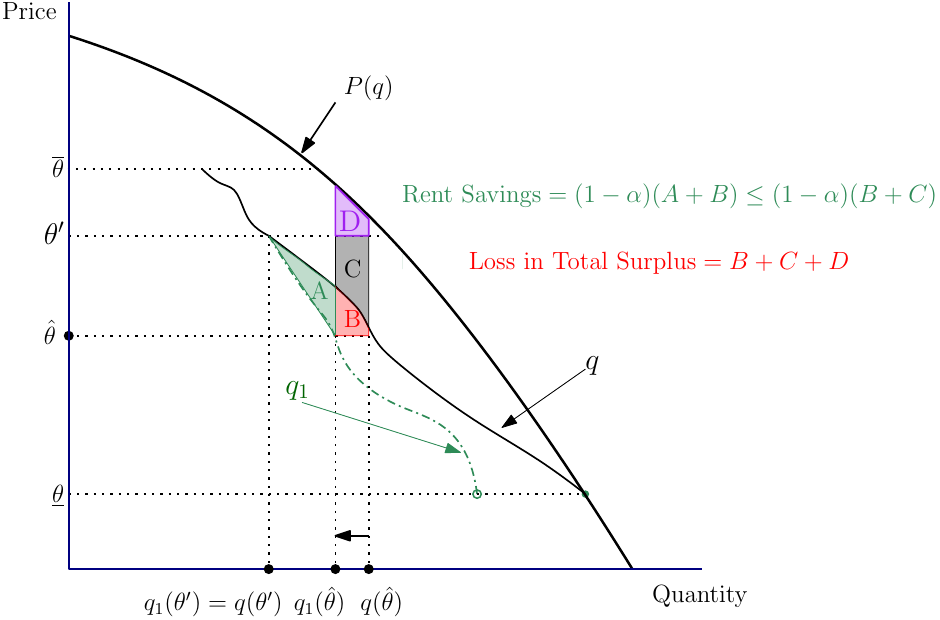}
    \caption{Strict DD and sufficiency}
    \label{fig:illu}
\end{figure}

The proof of Theorem \ref{theo:undsuff} is in Appendix \ref{app:characterization}. Here, we discuss the role of strict DD in the proof. Consider a quantity rule $q$, which satisfies strict DD and corresponds to a deterministic floor-randomized mechanism $M$. Such a $q$ is depicted in Figure \ref{fig:illu}. The price-axis also plots the types, whereas the horizontal axis plots quantity. To show that $q$ is undominated, it must be that any other mechanism cannot improve the regulator's surplus at some type without lowering it at other types. In particular, the quantity rule $q_1$, which is obtained by lowering $q$ at some types, should not be able to improve the regulator's surplus at $\hat{\theta}$. The rent saving corresponding to $q_1$ is given by $(1-\alpha)(A+B)$. Whereas the corresponding loss in the total surplus is given by $(B+C+D)$. The strict DD ensures that $D>0$, and therefore, even if $\alpha=0$ and area $A$ equals area $C$, the rent saving is always strictly smaller than the loss in the total surplus.\footnote{Observe, in the context of $\hat{\theta}$ shown in Figure \ref{fig:illu}, we always have $(A+B) \le (B+C)$. This is because $(A+B)$ is the area of a triangular region that has the same base and same height as the rectangle with area $(B+C)$.} Thus, $q_1$ can never dominate $q$.

\subsection{Discussion}
There is a gap between our necessary (Theorem \ref{theo:undnec}) and sufficient (Theorem \ref{theo:undsuff}) conditions. Moreover, none of these conditions depend on $\alpha$, the parameter that captures the regulator's distributional preferences. These observations raise several questions: First, what kind of undominated mechanisms satisfy the necessary conditions but not the sufficient conditions? Does the set of such undominated mechanisms vary with $\alpha$? Finally, is there a significant loss in terms of expected regulator's surplus if we restrict attention to undominated mechanisms that satisfy our sufficient conditions? This section answers these questions.

To answer the first question, consider an IC and IR mechanism $M_e \equiv (q_e,u_e)$ in which each type $\theta$ produces quantity $q_e(\theta)$ and $u_e(\overline{\theta})=0$. We call $M_e$ the {\sl efficient mechanism}. Clearly, the efficient mechanism is floor randomized because it is deterministic and $q_e(\theta) > \hat{q}$ for all $\theta$ (Observation \ref{ob:hatq_bounds}). Moreover, it is continuous, and satisfies DD. However, it violates strict DD, and therefore, Theorem \ref{theo:undsuff} is silent on whether $M_e$ is undominated. The next result shows that $M_e$ is dominated, no matter what $\alpha \in [0,1)$ is, when $P$ is twice continuously differentiable.

\begin{theorem}
    \label{theo:efficient}
    If $P$ is twice continuously differentiable, then the efficient mechanism is dominated.
\end{theorem}

Theorem \ref{theo:efficient} highlights the limits of our sufficient (Theorem \ref{theo:undsuff}) conditions in identifying undominated mechanisms. It exhibits that some mechanisms satisfy the necessary conditions and are still dominated.\footnote{It is possible to strengthen Theorem \ref{theo:efficient} further. A mechanism is called {\sl partially efficient} if the associated quantity rule equals the efficient quantity rule for some interval of types. The proof of Theorem \ref{theo:efficient} can be adapted to show that every partially efficient mechanism is dominated.}

We now focus attention on mechanisms that satisfy the necessary conditions but do not satisfy the sufficient conditions. The following proposition establishes that the set of such mechanisms shrinks as $\alpha$ decreases.

\begin{proposition}\label{prop:comp_stat}
    Suppose $\alpha, \alpha' \in (0,1)$ with $\alpha > \alpha'$. If $\mathcal{M}_{\alpha}$ and $\mathcal{M}_{\alpha'}$ are the set of undominated mechanisms with weights $\alpha$ and $\alpha'$ respectively, then 
    \begin{align*}
        \mathcal{M}_{\alpha'} \subseteq \mathcal{M}_{\alpha}.
    \end{align*}
\end{proposition}
\begin{proof}
Suppose $M \equiv (r, q,u) \in \mathcal{M}_{\alpha'}$. Assume for contradiction that $M \notin \mathcal{M}_{\alpha}$. Then it is dominated for weight $\alpha$ by a mechanism $M'$. By Theorem \ref{theo:undnec}, $M'$ must be floor randomized, left continuous and satisfy DD. The following theorem strengthens this further (proof is in Appendix \ref{app:implications}).
\begin{theorem}\label{theo:comp_stat}
Let $M \equiv (r, q, u)$ be a floor-randomized, left-continuous mechanism satisfying DD. If $M$ is dominated then there exists another floor-randomized, left-continuous mechanism $M'$ satisfying DD that dominates $M$,  and $q'(\theta)r'(\theta) \le q(\theta)r(\theta)$ for all $\theta$.
\end{theorem}

Theorem \ref{theo:comp_stat} strengthens part (2) in Theorem \ref{theo:undnec}. Therefore, for every $\theta$,
\begin{align*}
\textsc{RS}_{\alpha}(\theta,M') - \textsc{RS}_{\alpha}(\theta,M) &= r'(\theta)\left[V(q'(\theta))-c-\theta q'(\theta)\right] -r(\theta)\left[V(q(\theta))-c-\theta q(\theta)\right] \\
&~~\qquad\qquad- (1-\alpha) \int \limits_{\theta}^{\overline{\theta}}\big(q'(y)r'(y)-q(y)r(y)\big)dy \\
&\ge 0 
\end{align*}
with strict inequality holding for some $\theta \in \Theta$. The proof is completed by observing that if the inequality holds for $\alpha<1$, then it also holds for $\alpha'<\alpha$ as $q'(y)r'(y) \le q(y)r(y)$ for all $y$. Therefore, 
\begin{align*}
     \textsc{RS}_{\alpha'}(\theta,M') &\ge \textsc{RS}_{\alpha'}(\theta,M) \qquad \forall~\theta
\end{align*}
with strict inequality holding for some $\theta$. This contradicts the fact that $M \in \mathcal{M}_{\alpha'}$.
\end{proof}

Thus, as the regulator puts more weight on the profits of the monopolist, the set of undominated mechanisms expands. Equivalently, the set of associated undominated quantity rules also expands. 
However, it is worth noting that if $\alpha=1$, there is a ``discontinuity."  This is because the set of undominated quantity rules is a singleton for $\alpha=1$. To see why, observe that equation (\ref{eq:RS2}) implies that the regulator's surplus at type $\theta$ in a mechanism $(r,q,u)$ is $r(\theta)\textsc{TS}(\theta,q(\theta))$, which is uniquely maximized at $q_e(\theta)$ with $r(\theta)=1$.\footnote{$M_e$ is an undominated mechanism for $\alpha=1$. There are other undominated mechanisms. They all have the same quantity rule $q_e$, but differ from $M_e$ only in the value of $u_e(\bar{\theta})$. Because the regulator's surplus does not depend on the value of $u_e(\bar{\theta})$, any positive value for $u_e(\bar{\theta})$ ensures IR and defines a different undominated mechanism.}

Finally, to answer the third question, fix an undominated mechanism $M \equiv (r,q,u)$ that violates strict DD. Since $M$ is undominated, it is a floor-randomized mechanism satisfying DD and left-continuity. Let $M'\equiv (r,q',u')$ be another floor-randomized mechanism satisfying strict DD and left-continuity such that $q'$ is ``arbitrarily close" to $q$. Observe, $q'$ can be chosen arbitrarily close to $q$ because $M$ satisfies DD and $M'$ satisfies strict DD. By continuity of the regulator's surplus in $q$, the regulator's surplus under mechanism $M'$ also lies ``arbitrarily close" to that under $M$. This shows that although the set of undominated mechanisms varies with $\alpha$, it contains a strict subset invariant in $\alpha$, and it is almost without loss of generality to focus attention on this subset for a regulator who evaluates mechanisms by their expected regulator surplus.

\section{Implications for regulation}
\label{sec:implications}

This section discusses the implications of our findings for regulatory design. First, it identifies a max-min optimal mechanism for a regulator who entertains a set of priors over the marginal costs. Then, it shows that any undominated mechanism satisfying a mild continuity condition maximizes the regulator's surplus for some prior over marginal cost. Finally, it discusses why deterministic regulatory mechanisms are focal and identifies properties of optimal mechanisms (expected regulator surplus maximizing mechanism for a given prior) invariant to the details of the regulator's prior.

\subsection{Max-min optimal mechanism}
\label{sec:maxmin}
For every IC and IR mechanism $M$ and a prior $G$ over $\Theta$, define the expected regulator surplus as\begin{align*}\textsc{RS}_{\alpha}(M,G) &= \int \limits_{\underline{\theta}}^{\bar{\theta}} \textsc{RS}_{\alpha}(\theta,M) dG(\theta).\end{align*}
The optimal mechanism in \citet{BM82} maximizes the regulator's expected surplus for a given prior $G$ over the marginal cost. However, the regulator may face ambiguity regarding this distribution. For example, the regulatory body may consist of a committee of individuals, who may have different priors over the monopolist's marginal cost. 

To model this uncertainty, let $\mathcal{F}$ denote the set of all probability distributions over $\Theta$ and $\mathcal{G} \subseteq \mathcal{F}$ be the set of distributions that the regulator considers plausible. Suppose there exists $G^\star \in \mathcal{G}$ such that $G^\star \succ_{{\rm fosd}} G$ for every $G \in \mathcal{G}$, where $\succ_{{\rm fosd}}$ is the first-order stochastic domination relation. Below, we give examples of $\mathcal{G}$ where $G^\star$ exists. 
\begin{enumerate}
\item If $\mathcal{G}=\mathcal{F}$, then $G^\star$ is the distribution that puts all the probability mass at $\bar{\theta}$. 

\item If $\mathcal{G}$ is a subset of $\mathcal{F}$ such that for every $G, G' \in \mathcal{G}$, we have $\min(G,G') \in \mathcal{G}$, then there will exist a $G^\star$ that first-order stochastically dominates every distribution in $\mathcal{G}$.

\item If $\mathcal{G}$ is such that it contains $\underline{F}$ and $\overline{F}$ with $\underline{F}(\theta) \le \overline{F}(\theta)$ for all $\theta$. Further, $\mathcal{G}=\{G: \alpha\overline{F}(\theta)+(1-\alpha)\underline{F}(\theta)~\textrm{for some}~\alpha \in [0,1]\}$. Then, $G^\star=\underline{F}$. 
\end{enumerate}
For such $\mathcal{G}$, we solve for the  {\sl max-min optimal mechanism}.

\begin{defn}
    An IC and IR mechanism $M$ is {\bf max-min optimal} for $\mathcal{G}$ if for every IC and IR mechanism $M'$
    \begin{align*}
        \inf_{G \in \mathcal{G}} \textsc{RS}_{\alpha}(M,G) \ge \inf_{G \in \mathcal{G}} \textsc{RS}_{\alpha}(M',G).
    \end{align*}
\end{defn}

\begin{proposition}\label{prop:maxmin}
    The optimal mechanism for prior $G^\star$ is max-min optimal for $\mathcal{G}$.
\end{proposition}
If we further assume that $G^\star$ admits a continuous and positive density, then the optimal mechanism for $G^\star$ is the \cite{BM82} optimal mechanism. Consequently, max-min optimal mechanism is the \cite{BM82} optimal mechanism for $G^\star$.

The proof of Proposition \ref{prop:maxmin} in Appendix \ref{app:implications} is based on the following lemma which is of independent interest. It shows that for any floor-randomized mechanism satisfying DD is {\sl monotonic}: the expected regulator's surplus for prior $G$ is higher than that for prior $G'$, whenever $G' \succ_{{\rm fosd}} G$. Consequently, every undominated mechanism is also monotonic.\footnote{An example of a dominated mechanism that is not monotonic is available upon request.} This lemma leverages the unique feature of floor-randomized mechanisms that operation decisions are monotone functions.
\begin{lemma}\label{lem:mono}
    If $M \equiv (r, q, u)$ is a floor-randomized mechanism that satisfies DD, then for all $G, G' \in \mathcal{F}$ such that $G' \succ_{{\rm fosd}} G$,
    \begin{align*}
    \textsc{RS}_{\alpha}(M, G) \ge  \textsc{RS}_{\alpha}(M, G').
    \end{align*}
\end{lemma}
The proof of Lemma \ref{lem:mono} is in Appendix \ref{app:implications}. If we treat a distribution in $\mathcal{F}$ as a ``technology", then $G' \succ_{{\rm fosd}} G$ has the interpretation that $G$ is a better technology than $G'$ (because $G$ puts more probability mass on lower marginal costs). A consequence of Lemma \ref{lem:mono} is that as the technology gets better the regulator's expected surplus improves from a floor-randomized mechanism satisfying DD.

\subsection{Undomination and optimality}
\label{sec:expundom}

Our notion of undomination implies that when the regulator's prior has full support, an optimal mechanism that is undominated exists. This observation leads to a natural question: for a given undominated mechanism, does it maximize the regulator's expected surplus for some prior? We answer this question below.

Let $\mathcal{F}$ be the set of priors over $\Theta$ such that every $F\in \mathcal{F}$ has a bounded density function $f: \Theta \rightarrow \Re_+$ associated with it, where $\int \limits_{\underline{\theta}}^{\overline{\theta}} f(\theta)d\theta = 1$. For any  prior  $F \in \mathcal{F}$, the regulator's expected surplus from mechanism $M$ can be written as
\[
\textsc{RS}_\alpha(M,F) := \int \limits_{\underline{\theta}}^{\overline{\theta}} \textsc{RS}_\alpha(\theta, M) f(\theta) d\theta
\]
where $f$ is some density function associated with $F$.

\begin{proposition}
    \label{prop:max}
Suppose $M \equiv (r,q,u)$ is an undominated IC and IR mechanism such that $\textsc{RS}_{\alpha}(\theta,M)$ is continuous at $\underline{\theta}$. Then there exists $F \in \mathcal{F}$ such that
\begin{align}\label{eq:undomnew}
    \textsc{RS}_\alpha(M,F) &\ge \textsc{RS}_\alpha(M',F)
\end{align}
for every IC and IR mechanism $M'$.
\end{proposition}

Proposition \ref{prop:max} has two implications. First, a regulator using an undominated mechanism (with regulator surplus being continuous at the lowest type) can always come up with some prior and justify it as an optimal mechanism for that prior. Second, we can define a new notion of undomination where an IC and IR mechanism $M$ is undominated if there exists an $F$ such that (\ref{eq:undomnew}) holds for every IC and IR mechanism $M'$. Proposition \ref{prop:max} shows that if $M$ is undominated in our sense and the regulator surplus is continuous at the lowest type, it is also undominated in the new sense. 

The proof of Proposition \ref{prop:max} is in the Appendix \ref{app:implications}. \citet{MV07} prove a related result in the context of undominated mechanisms for a monopolist selling multiple goods to a buyer. Our proof builds on their separating hyperplane argument, however, we require two additional steps as the regulator's objective function is not linear in $q(\theta)$ for a given $\theta$.

\subsection{Deterministic regulatory mechanisms}

In this section, we use the properties of undominated mechanisms stated in Theorem \ref{theo:fr} and \ref{theo:undnec} to establish that focusing attention on deterministic mechanisms is without loss of generality in the following sense.
\begin{proposition}
    \label{prop:domdet}
    For every IC and IR mechanism $\widetilde{M}$ and every prior $G$ over $\Theta$, there exists a deterministic mechanism $M^\star$ such that 
    \begin{align*}
        \textsc{RS}_{\alpha}(M^\star,G) &\ge \textsc{RS}_{\alpha}(\widetilde{M},G).
    \end{align*}
\end{proposition}
To prove this, fix an IC and IR mechanism $\widetilde{M}$ and a prior $G$. First, either $\widetilde{M}$ is undominated or by Theorem \ref{theo:undnec}, there exists a floor-randomized mechanism $M$ that dominates $\widetilde{M}$. By Theorem \ref{theo:fr}, if $\widetilde{M}$ is undominated, it is also floor randomized. Hence, we conclude that there exists a floor-randomized mechanism $M$ (where $M$ may be same as $\widetilde{M}$) such that 
\begin{align}\label{eq:domdet1}
        \textsc{RS}_{\alpha}(M,G) &\ge \textsc{RS}_{\alpha}(\widetilde{M},G)
\end{align}
Now, let $M \equiv (r,q,u)$. Let $\mathcal{M}(q)$ be the set of all floor-randomized mechanisms for the quantity rule $q$:
\begin{align*}
    \mathcal{M}(q):=\{(r^\dag,q,u^\dag): (r^\dag,q,u^\dag)~\textrm{is floor randomized}\}.
\end{align*}
Clearly, $\mathcal{M}(q)$ is non empty because it contains $M$. In fact, it contains multiple elements when $q(\theta)=\hat{q}$ over an interval of $\theta$s. Observe that $\mathcal{M}(q)$ is a convex and compact set: if $(r^\dag,q,u^\dag)$ and $(r^\ddag,q,u^\ddag)$ are two mechanisms in $\mathcal{M}(q)$, their convex combination gives another floor-randomized mechanism. Hence, by Krien-Milman theorem, the set $\mathcal{M}(q)$ is convex hull of its {\sl extreme points}, where we define an extreme point as follows.
\begin{defn}
    A mechanism $(r^\dag,q,u^\dag) \in \mathcal{M}(q)$ is an {\bf extreme point} if there does not exist two other mechanisms $(r',q,u'), (r'',q,u'') \in \mathcal{M}(q)$ such that for some $\lambda \in (0,1)$
    \begin{align*}
        (r^\dag(\theta),u^\dag(\theta)) &= \lambda (r'(\theta),u'(\theta)) + (1-\lambda)(r''(\theta),u''(\theta))~\qquad~\forall~\theta \in \Theta.
    \end{align*}
\end{defn}

The next lemma shows that every extreme point of $\mathcal{M}(q)$ is a deterministic mechanism.

\begin{lemma}
    \label{lem:extreme}
    Every extreme point of $\mathcal{M}(q)$ is deterministic.
\end{lemma}
The proof of the above Lemma is similar to the proof of Lemma 4 in \citet{MV07} and is given in Appendix \ref{app:implications}.\footnote{Lemma 4 in \citet{MV07} shows that in a model where a seller is selling a single object to a single buyer, the extreme points of the set of IC mechanisms are deterministic mechanisms. In our model, this is not enough to guarantee the centrality of deterministic mechanism. Further, neither the set of undominated mechanisms nor the set of floor-randomized mechanisms form a convex set. However, fixing a floor-randomized mechanism $(r,q,u)$, the set $\mathcal{M}(q)$ is convex.} It shows that any mechanism $(r,q,u) \in \mathcal{M}(q)$ involving randomization, that is $r(\theta) \in (0,1)$ for some $\theta$, can be written as a convex combination of two mechanisms in the set $\mathcal{M}(q)$.

Thus, any objective function that is linear in $r$ will attain its maximum at a deterministic mechanism. In particular, fixing $q$, consider the following objective function
\[
\max_{(r,u)} \int \limits_{\Theta}\left(r(\theta)\textsc{TS}(\theta, q(\theta))-(1-\alpha)u(\theta)\right)dG(\theta).
\]
IC and IR implies $u(\theta)=u(\bar{\theta}) + \int \limits_{\theta}^{\bar{\theta}}q(y)r(y)dy$, where $u(\bar{\theta}) \ge 0$. Thus, to maximize the objective function, we can set $u(\bar{\theta})=0$ and solve
\[
\max_{r} \int \limits_{\Theta}\left(r(\theta)\textsc{TS}(\theta, q(\theta))-(1-\alpha)\int \limits_{\theta}^{\bar{\theta}}q(y)r(y)dy\right)dG(\theta).
\]This is clearly linear in $r$. As a result, Lemma \ref{lem:extreme} implies that a maximum over all $\mathcal{M}(q)$ can be achieved at a deterministic mechanism $M^\star$. Hence, there exists a deterministic mechanism $M^\star$ such that 
\begin{align*}
        \textsc{RS}_{\alpha}(M^\star,G) &\ge \textsc{RS}_{\alpha}(M,G)
\end{align*}
where we used the fact that $M \in \mathcal{M}(q)$. Using (\ref{eq:domdet1}), we get the desired inequality of Proposition \ref{prop:domdet}. The pictorial depiction of the discussion so far is given in Figure \ref{fig:domdet}.

\begin{figure}[!hbt]
    \centering
    \includegraphics[width=3in]{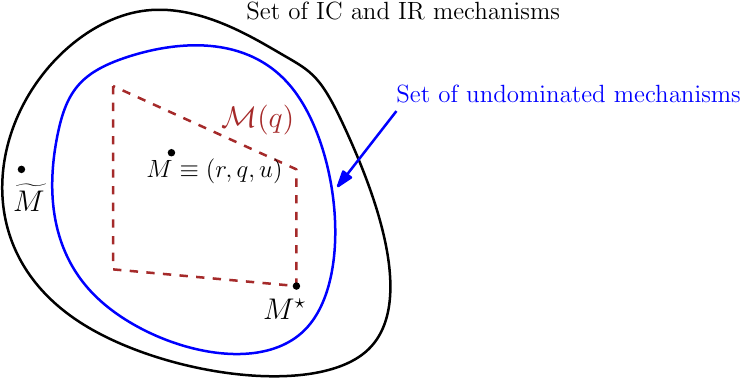}
    \caption{Pictorial depiction of Proposition \ref{prop:domdet}}
    \label{fig:domdet}
\end{figure}

An important implication of Proposition \ref{prop:domdet} concerns optimal regulation design for a given prior. 
We say a mechanism $M^{\rm{OPT}}\equiv (r^{\rm{OPT}},q^{\rm{OPT}},u^{\rm{OPT}})$ is {\sl optimal} for a prior $G$ over $\Theta$ if for every IC and IR mechanism $M$,
\begin{align*}
    \textsc{RS}_{\alpha}(M^{\rm OPT},G) &\ge \textsc{RS}_{\alpha}(M,G).
\end{align*}
We say $M^{\rm{OPT}}$ has a {\sl quantity floor} is $q^{\rm{OPT}}(\theta) \ge \hat{q}$ whenever $q^{\rm{OPT}}(\theta) > 0$. 
\begin{corollary}\label{cor:impli_opt}
For every prior, there exists a deterministic optimal mechanism that satisfies DD and has quantity floor.
\end{corollary}
Existence of deterministic optimal mechanism follows from Proposition \ref{prop:domdet}, which implies it is without loss of generality to focus attention on deterministic undominated mechanisms while searching for the optimal mechanism. The other two properties, DD and quantity floor, are guaranteed because every undominated mechanism has these properties. Corollary \ref{cor:impli_opt} identifies properties of optimal mechanisms that are invariant to the regulator's prior.\footnote{The optimal mechanism in \citet{BM82} is deterministic, satisfies DD, and has a quantity floor. However, they do not explicitly refer to the quantity floor, and their shut-down condition, when the cost structure is $c+ \theta q$, reduces to equation (\ref{eq:hatq}) that defines $\hat{q}$. In any optimal mechanism, the prior plays role only in determining the extent of downward distortion, and therefore, determines the smallest type which is allowed to operate and produce $\hat{q}$.} Moreover, it does not require any assumptions on the regulator's prior, and therefore, generalizes a finding of \citet{BM82}, which states that for any distribution that admits positive density, the optimal mechanism is deterministic.

\subsection{Comparative statics of changing demand and fixed cost}
\label{sec:cst}
This section examines how the set of deterministic undominated mechanisms changes with fixed cost and the inverse demand function. It focuses on the set of deterministic undominated mechanisms for two reasons: (a) this set exhibits sharp monotonicity properties with respect to changes in the exogenous variables, and (b) this set is pivotal for optimization over the class of undominated mechanisms (Proposition \ref{prop:domdet}).

It is almost immediate from equation (\ref{eq:hatq}) that, as the fixed cost $c$ increases, the quantity floor $\hat{q}$ increases, and therefore, the set of deterministic undominated mechanisms shrinks. A similar observation is true when the inverse demand curve $P$ is ``rotated" counter-clockwise around the point $(P^{-1}(\underline{\theta}),\underline{\theta})$.\footnote{It is necessary that the efficient quantity for the smallest marginal cost $\underline{\theta}$ is the same under two inverse demand functions, otherwise the corresponding sets of undominated mechanisms have an empty intersection.} Such rotations capture changes in the dispersion of consumer valuations due to product design and marketing strategies of the monopolist (see \cite{JM6}), and can be easily formalized for familiar parametric families of inverse demand curves as illustrated through the following examples.
\begin{enumerate}
    \item Linear inverse demand curves: Let $P(q) = a-bq$ with $b > 0$. A rotation of such inverse demand curve $P_\textrm{n}(q)=a_\textrm{n}-b_\textrm{n}q$, where subscript $\textrm{n}$ stands for $\textrm{n}$ew, is obtained by choosing $b_\textrm{n} < b$ and setting $a_\textrm{n}$ using the equation $b(a_\textrm{n}-\underline{\theta}) = b_\textrm{n}(a-\underline{\theta})$. 
    

    \item Inverse demand curves resulting from a discrete choice model: Consider an inverse demand curve $P(q) = V - \beta \ln \left(q/(1-q)\right)$, where $q \in (0,1)$. It is obtained from a discrete choice model where a consumer's utility from consuming the monopolist's product priced at $P$ is $V-P+\epsilon$ with $\epsilon$ being distributed according to the Logistic distribution with mean $0$ and scale parameter $\beta > 0$. If the consumer does not buy the monopolist's product, the consumer gets a utility of $0$. A rotation of such an inverse demand curve is given by $P_\textrm{n}(q) = V_\textrm{n} - \beta_\textrm{n} \ln \left(q/(1-q)\right)$, where $\beta_\textrm{n} < \beta$ and $V_\textrm{n}$ is determined by $\beta (V_\textrm{n}-\underline{\theta}) = \beta_\textrm{n}(V-\underline{\theta})$. 
\end{enumerate}

To state the result formally, let $\mathcal{D}(P)$ be the set of {\bf deterministic} undominated mechanisms when the inverse demand function is $P$.

\begin{proposition}
    \label{prop:pshrink}
    Suppose $P$ and $P_{{\rm n}}$ are strictly decreasing inverse demand functions satisfying [A1] and [A2]. Furthermore, for all $q$ and $q'< q$,
    \begin{align}
    \label{eq:decdiff}
        P(q) - P_{{\rm n}} (q) < P(q') - P_{{\rm n}} (q')
    \end{align}
    and $P_{{\rm n}}(P^{-1}(\underline{\theta})) = \underline{\theta}$. Then, $\mathcal{D}(P_{{\rm n}}) \subseteq \mathcal{D}(P)$.
\end{proposition}


\section{Related literature}
\label{sec:lit}

There is an extensive theoretical literature on mechanism design approach to regulation. \cite{AS06} and \cite{AS04} are excellent surveys, where the latter provides a unified treatment of various results in the literature. We work in the model of \cite{BM82}, which derives the optimal regulatory mechanism for a regulator who does not observe the marginal cost. The optimal regulatory mechanism depends on the prior. In contrast, our analysis is prior-free and seeks to understand undominated mechanisms.

While we, like \cite{BM82}, assume that the inverse demand function is known to the regulator, \cite{LS88AER} analyzes a model where the regulator observes the cost function of the monopolist but not the inverse demand function. \cite{LS88Rand, A99} analyze optimal regulatory mechanisms when the regulator does not know the cost and inverse demand functions.\footnote{\cite{D93,AR99} study the multidimensional versions of the problem. Using ideas from the delegation literature, \cite{AB22} derives optimal regulatory mechanisms in the single-dimensional case when the regulator cannot use transfers. \citet{YZ23} show that the optimal regulation for oligopolistic competition is a yardstick price cap mechanism. Also, there are papers that consider the moral hazard aspect of regulation. For instance, \cite{LT86} considers a model where the regulator observes the fixed cost and the marginal cost, but the monopolist can put unobservable effort into lowering these costs.}

Undominated mechanism design has been studied in other models of mechanism design. For instance, in the multidimensional screening problem (where a monopolist is selling multiple goods to a buyer with multidimensional private values), \cite{MV07} consider undominated selling mechanisms for the monopolist and show that they are optimal for {\sl some} prior distribution over values of the buyer. Unlike us, they do not give necessary and sufficient conditions for a mechanism to be undominated. 

For allocating a single object to multiple agents (with quasilinear utility), \cite{S13} characterizes undominated mechanisms in the class of strategy-proof, ex-post individually rational, feasible (weak budget-balance), anonymous, and envy-free mechanisms -- see similar characterizations in \cite{AV21} in the context of a public good problem. When selling multiple objects to a set of buyers who demand at most one object, \cite{KMS20} show that a generalization of the Vickrey-Clarke-Groves (VCG) mechanism is the unique undominated mechanism among the class of strategy-proof and ex-post individually rational mechanisms satisfying {\sl no-subsidy, no-wastage, and equal treatment of equals}.\footnote{The generalization of VCG is necessary because they allow agents to have non-quasilinear preferences.} As is apparent, this literature focuses on settings with many agents and characterizes undominated mechanisms under additional constraints than IC and IR. A recent work by \cite{BLW2024} analyzes undominated mechanisms in a single object private values auction model. It shows that any second-price auction with a deterministic positive reserve price is undominated, but any second-price auction with a non-degenerate random reserve price is dominated.

As alluded in the introduction, an arbitrary IC and IR regulatory mechanism need not be monotonic, that is, the expected regulator surplus may reduce even when the distribution over marginal cost gets better as per first-order stochastic dominance relation. However, undominated regulatory mechanisms, as shown by our Lemma \ref{lem:mono}, are always monotone. This idea of considering monotone mechanisms (i.e., designer's payoff increases as the prior distribution over the agent's valuation gets better) was first introduced in the mechanism design literature by \cite{HR15}. They show that deterministic and symmetric IC mechanisms are monotone if the monopolist is selling multiple goods, but every IC mechanism is monotone if the monopolist is selling a single good.

\newpage
\appendix

\section{Proofs of Section \ref{sec:floor_rand}}
\label{app:floor_rand}

\noindent {\sc \textbf{Proof of Theorem \ref{theo:fr}.}} Let $M \equiv (r,q,u)$ be an IC and IR mechanism, and $M^T \equiv (r^T,q^T,u^T)$ be its unique floor randomized transformation. Recall, for every $\theta$, we have $q^T(\theta)r^T(\theta)=q(\theta)r(\theta)$. Thus, using the notation $\overline{V}(q)=V(q)-c$, we can write
\begin{align*}
 \textsc{RS}_\alpha(\theta, M^T) & =   r^T(\theta)\overline{V}(q^T(\theta)) -\theta q^T(\theta)r^T(\theta) - (1-\alpha) u^T(\theta) \\
&= r(\theta)\overline{V}(q(\theta)) - \theta q(\theta)r(\theta) - (1-\alpha) \big(u(\theta)-u(\overline{\theta})\big) \\
&\qquad + \Big(r^T(\theta)\overline{V}(q^T(\theta)) - r(\theta)\overline{V}(q(\theta))\Big) \\
&= \textsc{RS}_\alpha(\theta, M) + (1-\alpha)u(\overline{\theta})+ \Big(r^T(\theta)\overline{V}(q^T(\theta)) - r(\theta)\overline{V}(q(\theta))\Big)
\end{align*}
Because $M$ is an IR mechanism, $u(\overline{\theta}) \ge 0$, and therefore,
\begin{align}\label{eq:diff}
\textsc{RS}_\alpha(\theta, M^T) - \textsc{RS}_\alpha(\theta, M) &\ge r^T(\theta)\Big(V(q^T(\theta))-c\Big) -r(\theta)\Big(V(q(\theta))-c\Big).
\end{align}
To complete the proof, it suffices to show that the RHS of inequality (\ref{eq:diff}) is non-negative for all $\theta$, and it is positive for some $\theta$. Before we show that,
consider the second term on the RHS of inequality (\ref{eq:diff}). If $q(\theta)=0$, then $r(\theta)=0$, and therefore, the whole term is zero. Thus, this term can be thought of as a convex combination of points on the graph of the function $\widetilde{V}(q)= \big(V(q) - c\big)\mathbf{1}_{q > 0}$. More formally, if $q(\theta)>0$, then
\[
r(\theta)\Big(V(q(\theta))-c\Big) = r(\theta)\widetilde{V}(q(\theta)) + (1-r(\theta))\widetilde{V}(0).
\]
A similar interpretation also applied to the first term on the RHS. Thus, to evaluate the RHS of inequality (\ref{eq:diff}), we need to consider the {\sl concave closure} of $\widetilde{V}(q)$, which is the smallest concave function weakly above $\widetilde{V}$.  

Observe $\widetilde{V}(q)$ is strictly increasing and strictly concave for all $0<q<\bar{q}$, where $\bar{q}$ is the unique $q$ such that $P(q)=0$. This is because the derivative of $\widetilde{V}(q)$ at $q > 0$ is $P(q)$ which is positive and strictly decreasing. However, $\widetilde{V}$ has a discontinuity at $q=0$ because $\widetilde{V}(0) = 0$ and $\lim_{q \downarrow 0}\widetilde{V}(q)=-c$. This means that $\widetilde{V}$ is not concave. The {\bf concave closure} of function $\widetilde{V}$ is denoted by $co(\widetilde{V})$ and is given by
\begin{align*}
co(\widetilde{V})(q) = \begin{cases}
    q P(\hat{q}) = \frac{\overline{V}(\hat{q})}{\hat{q}}q & \textrm{if} \ q  < \hat{q}\\
    \overline{V}(q) & \textrm{if} \ q  \ge \hat{q}
\end{cases}
\end{align*} 
Since $\hat{q}$ is such that $\overline{V}(\hat{q})=\hat{q}P(\hat{q})$, it follows that $co(\widetilde{V})$ defines the concave envelope of $\widetilde{V}$. Thus, $co(\widetilde{V})(q)$ is linearly increasing in $q$ till $\hat{q}$. For $q \ge \hat{q}$, $co(\widetilde{V})(q)$ is strictly concave and equals $\overline{V}(q) $. Figure \ref{fig:concav_0} in the main text graphically illustrates $co(\widetilde{V})$. 

Pick $\theta \in \Theta$. If $q(\theta)r(\theta)=0$, then the regulator's surplus is identical under both $M$ and $M^T$. Below, we consider two cases and argue that the RHS of the inequality (\ref{eq:diff}) is non-negative for all $\theta$.

\noindent \textsc{\textbf{Case 1}}. Let $\theta$ be such that $q(\theta)r(\theta) \ge \hat{q}$. Hence, by definition $r^T(\theta)=1$ and $q^T(\theta)=q(\theta)r(\theta)$. Then,
\begin{align*}
   \textsc{RS}_{\alpha}(\theta,M^T)-\textsc{RS}_{\alpha}(\theta,M) &\ge   r^T(\theta)\overline{V}(q^T(\theta)) - r(\theta)\overline{V}(q(\theta)) \\
    &= co(\widetilde{V})(q(\theta)r(\theta))-r(\theta)co(\widetilde{V})(q(\theta))\\
    &\ge 0
\end{align*}
where the inequality follows from concavity of $co(\widetilde{V})$ and $co(\widetilde{V})(0)=0$. If $r(\theta) < 1$, then inequality is strict because of the strict concavity of $co(\widetilde{V})$ on $q \ge \hat{q}$.

\noindent \textsc{\textbf{Case 2}}. Let $\theta$ be such that $0< q(\theta)r(\theta) < \hat{q}$. Then, $q^T(\theta)=\hat{q}$ and $r^T(\theta)=q(\theta)r(\theta)/\hat{q}$. Moreover,
\begin{align*}
   \textsc{RS}_{\alpha}(\theta,M^T)-\textsc{RS}_{\alpha}(\theta,M) &\ge   r^T(\theta)\overline{V}(q^T(\theta)) - r(\theta)\overline{V}(q(\theta)) \\
    &= \frac{q(\theta)r(\theta)}{\hat{q}}\overline{V}(\hat{q})-r(\theta)\overline{V}(q(\theta))\\
    &= co(\widetilde{V})(q(\theta)r(\theta))-r(\theta)\overline{V}(q(\theta))\\
    &\ge co(\widetilde{V})(q(\theta)r(\theta))-r(\theta)co(\widetilde{V})(q(\theta)) \\
    &\ge 0
\end{align*}
where the second equality follows from the definition of $co(\widetilde{V})$, the second inequality follows because $co(\widetilde{V})(q) \ge \overline{V}(q)$ for all $q$, and the last inequality follows from concavity of $co(\widetilde{V})$ and $co(\widetilde{V})(0)=0$. If $q(\theta) \neq \hat{q}$, this inequality is strict.

To complete the proof, observe that $M$ is not floor randomized, and therefore, we must have at least one $\theta$ where $q(\theta)r(\theta) \ge \hat{q}$ but $r(\theta) < 1$ or $0 < q(\theta)r(\theta) \le \hat{q}$ and $q(\theta) \neq \hat{q}$. That is, the RHS of inequality (\ref{eq:diff}) is positive at this $\theta$.{\hfill $\blacksquare{}$ \vspace{12pt}}

\section{Proofs of Section \ref{sec:characterization}}\label{app:characterization}

\noindent {\sc \textbf{Proof of Observation \ref{ob:hatq_bounds}.}} Let $\Delta(q) = \overline{V}(q) - P(q)q$. Notice that $\Delta(q)$ is strictly increasing and continuous in $q$ as $\Delta'(q) = -P'(q)q > 0$. Further, by definition of $\hat{q}$, we have $\Delta(\hat{q})=0$.
Now, 
\begin{align*}
\Delta(q_e(\bar{\theta})) & = \overline{V}(P^{-1}(\bar{\theta})) - \bar{\theta}P^{-1}(\bar{\theta}) >0
\end{align*}
where the inequality follows from assumption $[A2]$. By continuity and monotonicity of $\Delta$, we have $\hat{q} < q_e(\overline{\theta})$. This further implies $\hat{q} < q_e(\theta)$ for all $\theta$ because $q_e$ decreases in $\theta$. Finally, consider
\begin{align*}
    \textsc{TS}(\theta, \hat{q}) = \overline{V}(\hat{q})-\theta \hat{q} \ge \overline{V}(\hat{q})-\overline{\theta} \hat{q} = \big(P(\hat{q})-\overline{\theta}\big)\hat{q} >0
\end{align*}
where the last inequality holds because $\hat{q} < q_e(\overline{\theta})$ and $P(\cdot)$ is strictly decreasing.{\hfill $\blacksquare{}$ \vspace{12pt}}

\noindent {\sc \textbf{Proof of Theorem \ref{theo:undnec}.}} Suppose $M$ is undominated. By Theorem \ref{theo:fr}, it is floor randomized. Lemma \ref{lem:dd} and Lemma \ref{lem:lc}, which are stated and proved below, establish that a floor-randomized mechanism violating DD or left-continuity is dominated. Thus, completing the proof of Part 1.
     
As for Part 2, consider a dominated IC and IR mechanism $M$, which is not floor randomized. By Theorem \ref{theo:fr}, $M$ is dominated by its floor randomized transformation $M^T$. We can further use Lemma \ref{lem:dd} and Lemma \ref{lem:lc} to transform $M^T$ into another floor-randomized mechanism $\widetilde{M}$ that dominates $M$ and satisfies DD and left-continuity.{\hfill $\blacksquare{}$ \vspace{12pt}}

\begin{lemma} \label{lem:dd}
Suppose $M \equiv (r,q,u)$ is a floor-randomized mechanism. Then, there exists a floor-randomized mechanism $\widetilde{M} \equiv (\tilde{r},\tilde{q},\tilde{u})$ satisfying DD such that
    \begin{align} 
        \tilde{q}(\theta) &\le q(\theta) &\forall~\theta > \underline{\theta} \label{eq:uu3} \\
        \tilde{u}(\theta) &\le u(\theta) &\forall~\theta \nonumber \\
        \textsc{RS}_{\alpha} (\theta,\widetilde{M}) &\ge \textsc{RS}_{\alpha}(\theta,M) &\forall~\theta.\label{eq:uu4}
    \end{align}
Inequalities (\ref{eq:uu3}) and (\ref{eq:uu4}) are strict for all $\theta$ where $M$ violates DD. 
\end{lemma}
\begin{proof} Define a new quantity rule $\tilde{q}$ as follows 
    \begin{align*} 
        \tilde{q}(\theta) &:= 
        \begin{cases}
            q_e(\theta) & \textrm{if $\theta = \underline{\theta}$ or $q(\theta) > q_e(\theta)$} \\
            q(\theta) & \textrm{otherwise}   
        \end{cases}
    \end{align*} 
Consider mechanism $\widetilde{M} = (\tilde{r},\tilde{q},\tilde{u})$, where $\tilde{r}(\underline{\theta})=1$ and $\tilde{r}(\theta)=r(\theta)$ for all $\theta \neq \underline{\theta}$. Moreover,
\begin{align}\label{eq:def_tilde{u}} 
            \tilde{u}(\theta) &:= \int \limits_{\theta}^{\overline{\theta}} \tilde{q}(y)\tilde{r}(y)dy~\qquad~\forall~\theta. 
\end{align}
Both $q$ and $r$ are decreasing in $\theta$ because $M$ is a floor-randomized mechanism. Moreover, $q_e$ also decreases in $\theta$. Thus, $\tilde{q}$ as well as $\tilde{r}$ decrease in $\theta$. Monotonicity of $\tilde{q}(\theta)\tilde{r}(\theta)$ in $\theta$ along with equation (\ref{eq:def_tilde{u}}) implies $\widetilde{M}$ is IC. It is also IR with $\tilde{u}(\overline{\theta})=0$. To see why 
$\widetilde{M}$ is a floor-randomized mechanisms, consider the following two cases.\\
\noindent \textsc{\textbf{Case 1}}. Suppose $\Theta^M_1 = \emptyset$. In this case, $q(\theta) \in \{0, \hat{q}\}$ for all $\theta$. Thus, $\tilde{q}$ differs from $q$ only at $\underline{\theta}$, implying that $\Theta^{\widetilde{M}}_{1} = \{\underline{\theta}\} $. Furthermore,
\begin{align*}
\Theta^{\widetilde{M}}_{01} &= \begin{cases}
\Theta^M_{01} & \textrm{if $\Theta_{01}^M = \emptyset$} \\
\Theta^{M}_{01}\setminus\{\underline{\theta}\}   &\textrm{otherwise}
\end{cases} &
\Theta^{\widetilde{M}}_{0} &= \begin{cases}
\Theta^M_{0}\setminus \{\underline{\theta}\} & \textrm{if $\Theta^M_{01}=\emptyset$} \\
\Theta^M_{0} & \textrm{otherwise}
\end{cases} 
\end{align*}

\noindent \textsc{\textbf{Case 2}}. If $\Theta^M_1 \neq \emptyset$, then $\tilde{q}$ differs from $q$ only on a subset of $\Theta^M_{1}$. For every $\theta$ in this subset, $\tilde{q}(\theta) = q_e(\theta) > \hat{q}$, and therefore, $\tilde{r}(\theta)=1$. Consequently,
\begin{align*}
\Theta^{\widetilde{M}}_{1} & = \Theta^M_1 & \Theta^{\widetilde{M}}_{01} &= \Theta^M_{01} &\Theta^{\widetilde{M}}_{0} &= \Theta^M_0.
\end{align*}

We now argue that mechanism $\widetilde{M}$ satisfies DD and all the relevant inequalities in the statement of Lemma \ref{lem:dd} hold. Observe that $\tilde{q}(\theta) \le \min\{q_e(\theta),q(\theta)\}$ for all $\theta$, and thus, $\tilde{q}(\theta)$ satisfies DD. Moreover, $\tilde{q}(\theta) \le q(\theta) $ for all $\theta$ with a strict inequality where $q$ violates DD. We also have  
        \begin{align*} 
        \tilde{u}(\theta) = \int \limits_{\theta}^{\overline{\theta}} \tilde{q}(y)\tilde{r}(y)dy \le \int \limits_{\theta}^{\overline{\theta}} q(y)r(y)dy = u(\theta) \qquad \forall~\theta. \label{eq:u1}
        \end{align*}
        Thus, to complete the proof it suffice to show that the total surplus under $\widetilde{M}$ is larger than that under $M$ at every $\theta$. Total surplus is unchanged for $\theta$ if $\theta \ne \underline{\theta}$ and $q(\theta) \le q_e(\theta)$. For $\theta$ such that $\theta = \underline{\theta}$ or $q(\theta) > q_e(\theta)$, we have $\tilde{q}(\theta)=q_e(\theta)$. By definition of $q_e$, total surplus is higher in $\widetilde{M}$ than in $M$.
\end{proof}

\begin{lemma} 
\label{lem:lc}
Suppose $M \equiv (r,q,u)$ is a floor-randomized mechanism satisfying DD. Then, there exists another floor-randomized mechanism $\widetilde{M} \equiv (\tilde{r}, \tilde{q},u)$ satisfying DD and left-continuity such that 
\begin{align*}
    \tilde{q}(\theta)\tilde{r}(\theta) &\ge q(\theta)r(\theta) &\forall~\theta \\
    \textsc{CS}(\theta,\widetilde{M}) &\ge \textsc{CS}(\theta,M) &\forall~\theta\\
    \textsc{RS}_{\alpha} (\theta,\widetilde{M}) &\ge \textsc{RS}_{\alpha}(\theta,M)&\forall~\theta
\end{align*}
with equality holding for almost all $\theta$. Moreover, the above inequalities are strict at $\theta$s where $M$ violates left-continuity.
\end{lemma}
\begin{proof}
Because $q(\theta)r(\theta)$ is decreasing, it violates left-continuity at countably many types. Let $\Theta_d$ be the set of such types. We modify $q$ and $r$ at these types to obtain a new quantity rule $\tilde{q}$ and a new operation decision $\tilde{r}$.
\begin{align*}
    \tilde{q}(\theta) &= \lim_{z \uparrow \theta}q(z)  &
    \tilde{r}(\theta) &= \lim_{z \uparrow \theta}r(z) & \forall~\theta \in \Theta_d
\end{align*}
By construction, $\tilde{q}\tilde{r}$ is left continuous at all $\theta \in \Theta_d$, and as a result, it is also left continuous at all $\theta$. Therefore, the IC and IR mechanism $\widetilde{M} \equiv (\tilde{r},\tilde{q},\tilde{u})$, where $\tilde{u}(\theta) = \int \limits_{\theta}^{\overline{\theta}}\tilde{q}(y)\tilde{r}(y)dy$, is also left continuous. Moreover, $\widetilde{M}$ is floor randomized and satisfies DD because $M$ satisfies both these properties. Below, we argue that $\widetilde{M}$ dominates $M$.

For mechanism $\widetilde{M}$, the regulator's surplus at any type $\theta$ is 
\begin{align*}
    \tilde{r}(\theta)\textsc{TS}(\theta,\tilde{q}(\theta)) - (1-\alpha)\tilde{u}(\theta).
\end{align*}
Because $\tilde{q}(\theta)\tilde{r}(\theta)$ differs from $q(\theta)r(\theta)$ on measure zero types, we have $\tilde{u}(\theta)=u(\theta)$ for all $\theta$. Thus,
it is enough to show that \[
\tilde{r}(\theta)\textsc{TS}(\theta,\tilde{q}(\theta)) > r(\theta)\textsc{TS}(\theta,q(\theta))~~\forall~\theta\in \Theta_d
\] 
Fix $\theta\in \Theta_d$. First, suppose that $r(\theta)=0$. Then, $q(\theta)=0$. However, $\tilde{q}(\theta)\tilde{r}(\theta)>0$. Otherwise, $qr$ will be continuous at $\theta$, a contradiction. In particular, we have $\tilde{q}(\theta)>0$ and $\tilde{r}(\theta)>0$, and a larger expected total surplus under $\widetilde{M}$ at $\theta$. To complete the proof, consider the remaining case where $r(\theta)>0$. Because $M$ is floor randomized, $q(\theta) \ge \hat{q}$. Moreover, by construction, either $\tilde{q}(\theta) > q(\theta)$ or $\tilde{r}(\theta) > r(\theta)$. Because both $M$ and $\widetilde{M}$ satisfy DD, the increase in quantity or probability of operation increases the expected total surplus at $\theta$. 
\end{proof}

\noindent {\sc \textbf{Proof of Theorem \ref{theo:undsuff}.}} Suppose $M \equiv (q,r,u)$ is a floor randomized, left-continuous mechanism that satisfies strict DD. By contradiction, assume $M$ is dominated by another mechanism $\widetilde{M} \equiv (\tilde{r},\tilde{q},\tilde{u})$. Part (2) of Theorem \ref{theo:undnec} implies that we can find $\widetilde{M}$ such that it is floor randomized, left continuous, and satisfies DD. Let
\begin{align*}
    D_{qr} &:= \{\theta: \tilde{q}(\theta) \tilde{r}(\theta) < q(\theta)r(\theta)\}.
\end{align*}
By Lemma \ref{lem:Dqr} (which is stated and proved below after this proof), $D_{qr}$ has a positive measure. Next, define 
\begin{align*}
    \theta_h &:= \sup~\{\theta: \theta \in D_{qr}\}.
\end{align*}
$\theta_h$ is a finite real number because $D_{qr}$ is non-empty and bounded. Clearly, $\theta_h > \underline{\theta}$. 
In what follows, the proof obtains a contradiction by showing that there exists $\theta$, in a left neighbourhood of $\theta_h$, at which the regulator gets a higher surplus under mechanism $M$. To do so, it considers two cases: 1) $r(\theta_h) \in [0,1)$, and 2) $r(\theta_h)=1$.   

\noindent {\sc \textbf{Case 1}}. 
If $r(\theta_h) \in [0,1)$, then $q(\theta_h)r(\theta_h) < \hat{q}$ (since $M$ is a floor-randomized mechanism). There exists $\theta_{\ell} < \theta_h$, along with a corresponding left-neighbourhood $\mathcal{N} : = [\theta_{\ell}, \theta_h)$ of $\theta_h$,  such that 
\begin{align*}
\hat{q} >_{(a)} q(\theta)r(\theta) >_{(b)} \tilde{q}(\theta)\tilde{r}(\theta)\geq 0 ~\qquad~\forall~\theta\in\mathcal{N}.
\end{align*}
Inequality $(a)$ holds because $qr$ is  decreasing in $\theta$ and left continuous at $\theta_h$. As for inequality $(b)$, it follows from the definition of $\theta_h$. Thus, for every $\theta \in \mathcal{N}$, we have $q(\theta)=\hat{q}$ and $r(\theta)>0$. Furthermore, $r(\theta) > \tilde{r}(\theta)$. To see why, notice that 
$\hat{q}r(\theta) > \tilde{q}(\theta)\tilde{r}(\theta)$ and $\widetilde{M}$ being floor randomized imply $\tilde{q}(\theta) \in \{0,\hat{q}\}$. In both these cases, inequality $(b)$ implies $r(\theta) > \tilde{r}(\theta)$. Now, let
    \[
    \Delta := \sup_{y \in \mathcal{N}} r(y)-\tilde{r}(y)
    \]
which is a strictly positive. Thus, there exists $\theta \in \mathcal{N}$ such that $r(\theta) - \tilde{r}(\theta)=\Delta - \epsilon$, where $\epsilon < \Delta\frac{P(\hat{q})-P(q_e(\overline{\theta}))}{P(\hat{q})-P(q_e(\underline{\theta}))}$. Consider the difference
\begin{align*}
      \textsc{RS}_{\alpha}(\theta,M)-\textsc{RS}_{\alpha}(\theta,\widetilde{M}) &=~~~ r(\theta)\big(\overline{V}(\hat{q})-\theta\hat{q}\big)-\tilde{r}(\theta)\big(\overline{V}(\tilde{q}(\theta))-\theta\tilde{q}(\theta)\big)-(1-\alpha)\big(u(\theta)-\tilde{u}(\theta)\big) \\
      &\ge_{(a)} \Big(r(\theta)-\tilde{r}(\theta)\Big)\big(\overline{V}(\hat{q})-\theta\hat{q}\big)-(1-\alpha)\big(u(\theta)-\tilde{u}(\theta)\big) \\
      &=_{(b)} \Big(r(\theta)-\tilde{r}(\theta)\Big)\hat{q} \big(P(\hat{q})-\theta\big)-(1-\alpha)\big(u(\theta)-\tilde{u}(\theta)\big)
      \end{align*}
      where inequality $(a)$ follows from $\tilde{q}(\theta) \in \{0,\hat{q}\}$. As for equality $(b)$, it uses the fact that $\overline{V}(\hat{q})=\hat{q}P(\hat{q})$. To reach a contradiction and complete the proof, it suffices to show that the right hand side of the above equation is positive.
      \begin{align*}
          & \Big(r(\theta)-\tilde{r}(\theta)\Big)\hat{q} \big(P(\hat{q})-\theta\big)-(1-\alpha)\big(u(\theta)-\tilde{u}(\theta)\big) \\
          &= \Big(r(\theta)-\tilde{r}(\theta)\Big)\hat{q} \big(P(\hat{q})-\theta\big)- (1-\alpha) \Big( \hat{q}\int \limits_{\theta}^{\theta_h} \big(r(y)-\tilde{r}(y)\big)dy\Big)  - (1-\alpha)  \Big(\int \limits_{\theta_h}^{\overline{\theta}} \big(q(y)r(y)-\tilde{q}(y)\tilde{r}(y)\big)dy\Big)\\
          &~\qquad~\textrm{(since $r(y) > \tilde{r}(y) \ge 0$ for all $y \in [\theta, \theta_h$) and $u(\overline{\theta})=\tilde{u}(\overline{\theta})=0$)} \\
          &\ge \Big(r(\theta)-\tilde{r}(\theta)\Big)\hat{q} \big(P(\hat{q})-\theta\big)- (1-\alpha) \Big( \hat{q}\int \limits_{\theta}^{\theta_h} \big(r(y)-\tilde{r}(y)\big)dy\Big) \\
          & \textrm{(since $q(y)r(y) \le \tilde{q}(y)\tilde{r}(y)$ for all $y > \theta_h$)} \\
          &\ge \Big(r(\theta)-\tilde{r}(\theta)\Big)\hat{q} \big(P(\hat{q})-\theta\big)- \Big( \hat{q}\int \limits_{\theta}^{\theta_h} \big(r(y)-\tilde{r}(y)\big)dy\Big)~\qquad~\textrm{(since $r(y) > \tilde{r}(y)$ for all $y \in [\theta,\theta_h)$)} \\
          &\ge \Big(r(\theta)-\tilde{r}(\theta)\Big)\hat{q} \big(P(\hat{q})-\theta\big)-  \hat{q} \Delta (\theta_h - \theta) \\
          &= (\Delta - \epsilon)\hat{q} \big(P(\hat{q})-\theta\big)-  \hat{q} \Delta (\theta_h - \theta) \\
          &= \Delta \hat{q} \big(P(\hat{q}) - \theta_h) - \epsilon \hat{q} \big(P(\hat{q})-\theta\big) \\
          &\ge \Delta \hat{q} \big(P(\hat{q}) - \bar{\theta}) - \epsilon \hat{q} \big(P(\hat{q})-\underline{\theta}\big) \\
          &= \Delta \hat{q} \big(P(\hat{q}) - P(q_e(\bar{\theta}))\big) - \epsilon \hat{q} \big(P(\hat{q})-P(q_e(\underline{\theta}))\big) \\
          &> 0
      \end{align*}
      where the last inequality follows from Lemma \ref{ob:hatq_bounds} ($\hat{q} < q_e(\bar{\theta})$) and the fact that $\epsilon<\Delta\frac{P(\hat{q})-P(q_e(\overline{\theta}))}{P(\hat{q})-P(q_e(\underline{\theta}))}$. \\
    
    \noindent \textsc{\textbf{Case 2}}. Suppose $r(\theta_h)=1$. Then, by Lemma \ref{lem:theta_h_equality} (which is stated and proved below after this proof), $\tilde{r}(\theta_h)=1$. Because both $M$ and $\widetilde{M}$ are floor randomized, $r(\theta)=\tilde{r}(\theta)=1$ for all $\theta \le \theta_h$. Therefore, it suffices to focus on $q$ and $\tilde{q}$, which are left continuous on $\theta \le \theta_h$,  for comparing the regulator's surplus for types below $\theta_h$.
    
    There exists a non-empty left-neighbourhood $\mathcal{N} := [\theta_l,\theta_h)$ of $\theta_h$ such that 
    \[
    \tilde{q}(\theta) <_{(a)} q(\theta) <_{(b)} q_e(\theta_h) ~\forall~\theta\in \mathcal{N}.
    \] 
     Inequality $(a)$ holds because of the definition of $\theta_h$. Whereas, inequality $(b)$ is guaranteed because $q(\theta_h) <q_e(\theta_h)$ ($M$ satisfies strict DD) and $q$ is left continuous and decreasing. Let
    \[
    \Delta = \sup_{y \in \mathcal{N}} q(y)-\tilde{q}(y)
    \]
    which is strictly positive. Thus, there exists $\theta \in \mathcal{N}$ such that $q(\theta) - \tilde{q}(\theta)=\Delta - \epsilon$, where $\epsilon < \Delta\frac{P(q(\theta_l))-P(q_e(\theta_h))}{P(q(\theta_l))-P(q_e(\underline{\theta}))}$. Now, consider the difference
    \begin{align*}
      \textsc{RS}_{\alpha}(\theta,M)-\textsc{RS}_{\alpha}(\theta,\widetilde{M}) &= \int \limits_{\tilde{q}(\theta)}^{q(\theta)}\big(P(z)-\theta\big)dz-(1-\alpha)\big(u(\theta)-\tilde{u}(\theta)\big).
      \end{align*}
To reach a contradiction and complete the proof, it suffices to show that the right hand side of the above equation is positive.
      \begin{align*}
     &\int \limits_{\tilde{q}(\theta)}^{q(\theta)}\big(P(z)-\theta\big)dz - (1-\alpha) \Big(u(\theta)-\tilde{u}(\theta)\Big)&\\
      &= \int \limits_{\tilde{q}(\theta)}^{q(\theta)}\big(P(z)-\theta\big)dz - (1-\alpha) \Big( \int \limits_{\theta}^{\theta_h} \big(q(y)-\tilde{q}(y)\big)dy\Big)& \\
      &\qquad - (1-\alpha) \Big( \int \limits_{\theta_h}^{\overline{\theta}} \big(q(y)r(y)-\tilde{q}(y)\tilde{r}(y)\big)dy\Big)&\textrm{(since $u(\overline{\theta})=0$ and $\tilde{u}(\overline{\theta}) = 0$)} \\
      &\ge \int \limits_{\tilde{q}(\theta)}^{q(\theta)}\big(P(z)-\theta\big)dz - (1-\alpha) \Big( \int \limits_{\theta}^{\theta_h} \big(q(y)-\tilde{q}(y)\big)dy\Big)&\textrm{(since $q(y)r(y) \le \tilde{q}(y)\tilde{r}(y)$ for all $y > \theta_h$)} \\
      &\ge \int \limits_{\tilde{q}(\theta)}^{q(\theta)}\big(P(z)-\theta\big)dz - \int \limits_{\theta}^{\theta_h} \big(q(y)-\tilde{q}(y)\big)dy&\textrm{(since $q(y)>\tilde{q}(y)$ for all $y \in [\theta,\theta_h)$)}\\
      &\ge \big(P(q(\theta_{\ell})) - \theta\big) (q(\theta) - \tilde{q}(\theta)) - \Delta (\theta_h-\theta)&\textrm{(since $q(\theta_{\ell}) \ge q(\theta)$ and by definition of $\Delta$)} \\
      &= \big(P(q(\theta_{\ell})) - \theta\big) (\Delta - \epsilon) - \Delta (\theta_h - \theta)& \\
      &= \Delta\big(P(q(\theta_{\ell})) - \theta_h\big) - \epsilon \big(P(q(\theta_{\ell})) - \theta\big) &\\
      &\ge \Delta\big(P(q(\theta_{\ell})) - \theta_h\big) - \epsilon \big(P(q(\theta_{\ell})) - \underline{\theta}\big) &\\
      &> 0&
      \end{align*}
where the last inequality follows from $\epsilon < \Delta\frac{P(q(\theta_l))-P(q_e(\theta_h))}{P(q(\theta_l))-P(q_e(\underline{\theta}))}$.{\hfill $\blacksquare{}$ \vspace{12pt}}

\begin{lemma}\label{lem:Dqr}
The set $D_{qr}$ has a positive measure.    
\end{lemma}
\begin{proof}
Both $M$ and $\widetilde{M}$ satisfy DD, and therefore, $q(\underline{\theta})=\tilde{q}(\underline{\theta})=q_e(\underline{\theta})$. Thus, the total surplus from both the mechanisms at $\underline{\theta}$ is the same. Because $\widetilde{M}$ dominates $M$, it must be that the rent paid by the monopolist of type $\underline{\theta}$ is lower in $\widetilde{M}$ than in $M$, that is,
\begin{align}\label{eq:rent}
    u(\underline{\theta}) = \int \limits_{\underline{\theta}}^{\bar{\theta}} q(z)r(z) dz \ge \int \limits_{\underline{\theta}}^{\bar{\theta}} \tilde{q}(z)\tilde{r}(z)dz = \tilde{u}(\underline{\theta}).
\end{align}
If the inequality is an equality, $q(\theta)r(\theta)=\tilde{q}(\theta)\tilde{r}(\theta)$ for almost all $\theta$, and by left-continuity of $M$ and $\widetilde{M}$, $q(\theta)r(\theta)=\tilde{q}(\theta)\tilde{r}(\theta)$ for all $\theta$. Because both $M$ and $\widetilde{M}$ are floor-randomized mechanisms, this implies $M = \widetilde{M}$, which is a contradiction. Thus, inequality (\ref{eq:rent}) is strict, and $q(\theta)r(\theta) > \tilde{q}(\theta)\tilde{r}(\theta)$ for a positive measure of $\theta$. Consequently, $D_{qr}$ has a positive measure.
\end{proof}

\begin{lemma}\label{lem:theta_h_equality}
If $r(\theta_h)=1$, then $\tilde{r}(\theta_h)=1$ and $q(\theta_h)=\tilde{q}(\theta_h)$.    
\end{lemma}
\begin{proof}
For all $\theta > \theta_h$, we have $q(\theta)r(\theta) \le \tilde{q}(\theta)\tilde{r}(\theta)$. Therefore, the rent paid by the monopolist of type  $\theta_h$ is lower in $M$ than in $\widetilde{M}$, that is $u(\theta_h) \le \tilde{u}(\theta_h)$. It must be that the expected total surplus is higher in $\widetilde{M}$ than $M$ at $\theta_h$:
\begin{align}\label{eq:ts1}
    \tilde{r}(\theta_h) \Big(\overline{V}(\tilde{q}(\theta_h)) - \theta_h\tilde{q}(\theta_h) \Big) \ge 
    r(\theta_h) \Big(\overline{V}(q(\theta_h)) - \theta_hq(\theta_h) \Big).
\end{align}
Otherwise, $\widetilde{M}$ cannot dominate $M$ at $\theta_h$, a contradiction. We now use this inequality to show the desired result. 

Suppose $r(\theta_h)=1$. In this case, definition of $\theta_h$ and left-continuity of $qr$ and $\tilde{q}\tilde{r}$ implies $q(\theta_h) \ge \tilde{r}(\theta_h)\tilde{q}(\theta_h)$. If $\tilde{r}(\theta_h) < 1$, then $\tilde{q}(\theta_h) \le \hat{q} \le q(\theta_h)$, which contradicts with inequality (\ref{eq:ts1}). Thus, $\tilde{r}(\theta_h) = 1$, and therefore,  $q(\theta_h) \ge \tilde{q}(\theta_h)$. However, inequality (\ref{eq:ts1}) along with the fact that $r(\theta_h) = \tilde{r}(\theta_h) = 1$ implies $\tilde{q}(\theta_h) \ge q(\theta_h)$. This proves $q(\theta_h)=\tilde{q}(\theta_h)$.
\end{proof}


\noindent {\sc \textbf{Proof of Theorem \ref{theo:efficient}.}} Let $Q=[q_e(\overline{\theta}),q_e(\underline{\theta})]$. We begin by arguing that there exists an interval $[q_0,q_1] \subset Q$ on which $P$ is either concave or convex. If $P''(q) = 0$ for all $q \in Q$, then any closed interval contained in $Q$ is the desired interval. Otherwise,  $P''(q) \neq 0$ for some $q$. By continuity of the second derivative, we can find a closed interval $[q_0,q_1]$ containing $q$ where $P$ is either concave or convex. Let $\theta_{\ell} = P^{-1}(q_1)$ and $\theta_h = P^{-1}(q_0)$.

Because $P$ is either concave or convex on $[q_0,q_1]$, it is differentiable almost everywhere. Let $-g(y)$, where $g(y) > 0$ for all $y \in [q_0,q_1]$, be a subgradient of $P$ at $y$ (it is equal to the derivative of $P$ almost everywhere). Then, we can write 
\begin{align*}
 P(q) &= P(q')-\int \limits_{q'}^{q}g(y)dy   & \forall q,q'.
\end{align*}
In the interval $[q_0,q_1]$, if $P$ is concave then $g$ is increasing. On the other hand, if $P$ is convex on $[q_0,q_1]$, then $g$ is decreasing. Consider the following quantity rule
 \[
    \tilde{q}(\theta) = \begin{cases} q_e(\theta) & \textrm{if $\theta \notin [\theta_{\ell},\theta_h]$ }\\
    q_e(\theta)-\epsilon(\theta_h-\theta) & \textrm{if $\theta \in [\theta_{\ell},\theta_h]$ }
    \end{cases}
    \]  
    where $\epsilon < \frac{1-\alpha}{2g(P^{-1}(\theta^\star))}$ with $\theta^\star=\theta_{\ell}$ if $P$ is concave and $\theta^\star=\theta_h$ if $P$ is convex. First, we assume that $P$ is concave on $[q_0,q_1]$. Step 1 shows that $\tilde{q}$ is decreasing, and therefore, $\widetilde{M}\equiv (\tilde{q},\tilde{u})$, where $\tilde{u}(\theta)=\int \limits_{\theta}^{\overline{\theta}}\tilde{q}(y)dy$ is an IC and IR mechanism. Step 2 establishes that $\widetilde{M}$ dominates $M_e$. Finally, Step 3 considers the case of $P$ being convex and highlights changes to Steps 1 and 2. 

    \noindent {\sc \textbf{Step 1.}} Observe $\tilde{q}$ is decreasing on $\Theta \setminus [\theta_{\ell},\theta_h]$ as $q_e$ is decreasing (IC of $M_e$). Thus, it suffices to show that $\tilde{q}$ is decreasing in the interval $[\theta_{\ell},\theta_h]$. To do so, note that $P$ is decreasing and concave, and hence, $P^{-1}$ is decreasing and concave on $[\theta_{\ell},\theta_h]$. As a result, $q_e$ is also decreasing and concave on $[\theta_{\ell},\theta_h]$. Thus, for every $\theta,\theta' \in [\theta_{\ell},\theta_h]$, we have $q_e(\theta')=q_e(\theta)-\int \limits_{\theta}^{\theta'}w(y)dy$, where $w(y)>0$ and $w$ is increasing. Moreover, $w(y)=1/g(P^{-1}(y))$ whenever $P(\cdot)$ is differentiable in $[\theta_{\ell},\theta_h]$. For any $\theta', \theta \in [\theta_{\ell},\theta_h]$ such that $\theta'>\theta$, we have
    \begin{align}
        \tilde{q}(\theta)-\tilde{q}(\theta') =q_e(\theta)-q_e(\theta')-\epsilon(\theta'-\theta) &=\int \limits_{\theta}^{\theta'}w(y)dy-\epsilon(\theta'-\theta)& \nonumber\\
        &\ge (w(\theta_{\ell})-\epsilon)(\theta'-\theta) > 0& \label{eq:mono}
    \end{align}
    where the first inequality follows since $w$ is increasing and the last inequality holds because $\epsilon < \frac{1-\alpha}{2g(P^{-1}(\theta_{\ell}))} < w(\theta_{\ell})$. \\

    \noindent {\sc \textbf{Step 2.}} Now we show that mechanism $\widetilde{M}$ dominates mechanism $M_e$. Since $q_e(\theta) \ge \tilde{q}(\theta)$ for all $\theta$, we get $\tilde{u}(\theta) \le u_e(\theta)$. For all $\theta \notin [\theta_{\ell},\theta_h]$, since $q_e(\theta)=\tilde{q}(\theta)$,  we immediately conclude that $\textsc{RS}_{\alpha}(\theta,\widetilde{M}) \ge \textsc{RS}_{\alpha}(\theta,M_e)$. For $\theta \in [\theta_{\ell},\theta_h]$, the reduction in the total surplus is
    \begin{align}
        &\textsc{TS}(\theta,M_e) - \textsc{TS}(\theta,\widetilde{M}) &\nonumber\\
        &= \int \limits_{\tilde{q}(\theta)}^{q_e(\theta)} \Big(P(z) - \theta \Big) dz &\nonumber\\
        &=  \int\limits_{\theta}^{P(\tilde{q}(\theta))} (P^{-1}(y) - \tilde{q}(\theta))dy &\textrm{(integrating along the price axis)} \nonumber\\
        &\le \quad  \int\limits_{\theta}^{P(\tilde{q}(\theta))} (q_e(y) - \tilde{q}(y))dy &\textrm{($\tilde{q}$ is decreasing and $q_e(y)=P^{-1}(y)$)} \nonumber\\
        &=  \epsilon \int\limits_{\theta}^{P(\tilde{q}(\theta))} (\theta_h -y)dy & \textrm{(by definition of $\tilde{q}$)}\nonumber\\
        &= \epsilon (P(\tilde{q}(\theta))-\theta)\Big((\theta_h-\theta)-\frac{(P(\tilde{q}(\theta))-\theta)}{2}\Big) & \nonumber\\
        &\le \epsilon (P(\tilde{q}(\theta))-\theta)(\theta_h-\theta) & \textrm{(since $P(\tilde{q}(\theta)) \ge P(q_e(\theta))=\theta$).} \label{eq:su1}
    \end{align}
Using increasingness of $g$ and the fact that $\theta=P(q_e(\theta))$, we get 
\begin{align*}
    P(\tilde{q}(\theta)) - \theta &= P(\tilde{q}(\theta)) - P(q_e(\theta)) = \int \limits_{\tilde{q}(\theta)}^{q_e(\theta)} g(y)dy \le g(q_e(\theta))(q_e(\theta) - \tilde{q}(\theta)) = \epsilon g(q_e(\theta)) (\theta_h - \theta).
    \end{align*}
Substituting this in equation (\ref{eq:su1}), we get 
\begin{align}
        \textsc{TS}(\theta,M_e) - \textsc{TS}(\theta,\widetilde{M}) &\le  \epsilon^2 g(q_e(\theta)){(\theta_h-\theta)}^2 \nonumber\\
        &\le  \epsilon^2 g(q_e(\theta_{\ell})){(\theta_h-\theta)}^2 \qquad~\textrm{(since $g$ is increasing and $q_e$ is decreasing)} \nonumber\\
        &=  \epsilon^2 g(P^{-1}(\theta_{\ell})){(\theta_h-\theta)}^2\qquad~\textrm{(since $q_e(\theta)=P^{-1}(\theta)$).} \label{eq:surplus}
    \end{align}
Next, reduction in rent paid at $\theta \in [\theta_{\ell},\theta_h]$ is
    \begin{align*}
       u_e(\theta) - \tilde{u}(\theta) &= \int \limits_{\theta}^{\overline{\theta}} (q_e(y) - \tilde{q}(y))dy = \epsilon \int \limits_{\theta}^{\theta_h} (\theta_h-y)dy = \frac{1}{2}\epsilon (\theta_h - \theta)^2.
    \end{align*}
    Hence, for any $\theta \in[\theta_{\ell},\theta_h]$, we get 
    \begin{align*}
        \textsc{RS}_{\alpha}(\theta,\widetilde{M}) - \textsc{RS}_{\alpha}(\theta,M_e) &= \textsc{TS}(\theta,\widetilde{M}) - \textsc{TS}(\theta,M_e) - (1-\alpha)(\tilde{u}(\theta) - u_e(\theta)) \\
        &\ge (1-\alpha) \frac{1}{2}\epsilon (\theta_h - \theta)^2 -  \epsilon^2 g(P^{-1}(\theta_{\ell})){(\theta_h-\theta)}^2 \\
        &= \frac{1}{2} \epsilon (\theta_h - \theta)^2 (1-\alpha - 2\epsilon g(P^{-1}(\theta_{\ell}))).
    \end{align*}
 Since $\epsilon < \frac{1-\alpha}{2g(P^{-1}(\theta_{\ell}))}$, we see that the last expression is positive except for $\theta=\theta_h$, where it is zero. This establishes that $\widetilde{M}$ dominates $M_e$. \\

    \noindent {\sc \textbf{Step 3.}} To complete the proof we consider the case that $P$ is convex on $[q_0,q_1]$. In this case, the proof is similar to the case when $P$ is concave. Below, we highlight main changes. First, we choose $\epsilon < \frac{1-\alpha}{2g(P^{-1}(\theta_h))}$.
    
    Since $g$ and $w$ are decreasing and $\epsilon < \frac{1-\alpha}{2g(P^{-1}(\theta_h))} < \frac{1}{g(P^{-1}(\theta_h))} = w(\theta_h)$, inequality (\ref{eq:mono}) in Step 1 reduces to
    \[
    \tilde{q}(\theta)-\tilde{q}(\theta') \ge (w(\theta_h)-\epsilon)(\theta'-\theta) > 0.
    \]
    
   As for Step 2, identical simplifications lead to the counterpart of inequality (\ref{eq:surplus}) by using $g$ is decreasing
   \[
    \textsc{TS}(\theta,M_e) - \textsc{TS}(\theta,\widetilde{M}) \le \epsilon^2 g(P^{-1}(\theta_h)){(\theta_h-\theta)}^2.
   \]
   The proof is completed by comparing $\textsc{RS}_{\alpha}(\theta,\widetilde{M})$ and $\textsc{RS}_{\alpha}(\theta,M_e)$. {\hfill $\blacksquare{}$ \vspace{12pt}}

\noindent {\sc \textbf{Proof of Theorem \ref{theo:comp_stat}.}}
If $M$ is dominated, by part 2 of Theorem \ref{theo:undnec}, there exists another floor-randomized mechanism  $\widetilde{M} \equiv (\tilde{r},\tilde{q},\tilde{u})$ satisfying DD and left-continuity such that 
\begin{align*}
    \textsc{RS}_{\alpha}(\theta,\widetilde{M}) \ge \textsc{RS}_{\alpha}(\theta,M)~\qquad~\forall~\theta \in \Theta
\end{align*}
with strict inequality for some $\theta$. Define 
\begin{align*} 
D_{qr} &:= \{\theta: \tilde{q}(\theta)\tilde{r}(\theta) < q(\theta)r(\theta)\}, & D_{\tilde{q}\tilde{r}} &:= \{\theta: \tilde{q}(\theta)\tilde{r}(\theta) > q(\theta)r(\theta)\}.
\end{align*}
By Lemma \ref{lem:Dqr}, set $D_{qr}$ has a positive measure. If $D_{\tilde{q}\tilde{r}}$ has zero measure, then left-continuity of $q(\theta)r(\theta)$ and  $\tilde{q}(\theta)\tilde{r}(\theta)$ implies that $\tilde{q}(\theta)\tilde{r}(\theta) \le q(\theta)r(\theta)$ for all $\theta$. Thus, setting $M' = \widetilde{M}$ completes the proof. For the rest of the proof, we assume $D_{\tilde{q}\tilde{r}}$ has a positive measure. Consider the following mechanism $M' \equiv (r', q',u')$:
\begin{align*}
r'(\theta) &:= \min(r(\theta),\tilde{r}(\theta))~\qquad~\forall~\theta \\
q'(\theta) &:= \min(q(\theta),\tilde{q}(\theta))~\qquad~\forall~\theta \\
u'(\theta) &:= \int \limits_{\theta}^{\overline{\theta}} q'(y)r'(y)dy~\qquad~\forall~\theta.
\end{align*}
Clearly, $M'$ is floor randomized and satisfies DD because both $M$ and $\widetilde{M}$ have these properties. To see why $q'(\theta)r'(\theta)$ is left continuous, observe that $q(\theta)r(\theta)$ and $\tilde{q}(\theta)\tilde{r}(\theta)$ are left continuous, and $q'(\theta)r'(\theta) = \min(q(\theta)r(\theta),\tilde{q}(\theta)\tilde{r}(\theta))$ for all $\theta$. By definition, $u'(\theta) \le \min(u(\theta),\tilde{u}(\theta))$ for all $\theta$. Moreover, $q'(\theta)r'(\theta) \le q(\theta)r(\theta)$ for all $\theta$ with a strict inequality on the set $D_{qr}$. We now show that $\textsc{RS}_{\alpha}(\theta,M') \ge \textsc{RS}_{\alpha}(\theta,M)$ for all $\theta$ with a strict inequality holding on $D_{\tilde{q}\tilde{r}}$. To do so, consider the two possible cases. \\

\noindent {\sc \textbf{Case 1.}} Suppose $\theta \in D_{\tilde{q}\tilde{r}}$. Thus, $\tilde{q}(\theta)\tilde{r}(\theta) > q(\theta)r(\theta)$. Combining this inequality with the fact that both $M$ and $\widetilde{M}$ are floor-randomized mechanisms, we have $q'(\theta)=q(\theta) \le \tilde{q}(\theta)$ and $r'(\theta) = r(\theta) \le \tilde{r}(\theta)$ with one of the inequalities being strict.
\begin{align*} 
\textsc{RS}_{\alpha}(\theta,M') &=\quad r'(\theta)\Big(\overline{V}(q'(\theta)) - \theta q'(\theta) \Big) - (1-\alpha) u'(\theta) &\\
&>_{(b)} r(\theta)\Big(\overline{V}(q(\theta)) - \theta q(\theta) \Big) - (1-\alpha) u(\theta)& \\
&=\quad \textsc{RS}_{\alpha}(\theta,M).
\end{align*}
The strict inequality (b) follows from Lemma \ref{lem:D_1}, which states $\tilde{u}(\theta) < u(\theta)$ for all $\theta \in D_{\tilde{q}\tilde{r}}$, and since $u'(\theta) \le \tilde{u}(\theta)$, we get $u'(\theta) < u(\theta)$ for all $\theta \in D_{\tilde{q}\tilde{r}}$.\footnote{Weaker version of inequality (b) follows immediately because $u'(\theta) \le u(\theta)$.} \\

\noindent {\sc \textbf{Case 2.}} Suppose $\theta \notin D_{\tilde{q}\tilde{r}}$. Then, just as in {\sc Case 1}, we can argue that $q'(\theta)=\tilde{q}(\theta)$ and $r'(\theta) = \tilde{r}(\theta)$.
\begin{align*} 
\textsc{RS}_{\alpha}(\theta,M') &= r'(\theta)\Big(\overline{V}(q'(\theta)) - \theta q'(\theta) \Big)- (1-\alpha) u'(\theta)& \\
&= \tilde{r}(\theta)\Big(\overline{V}(\tilde{q}(\theta)) - \theta \tilde{q}(\theta) \Big) - (1-\alpha) u'(\theta)&\\
&\ge \tilde{r}(\theta)\Big(\overline{V}(\tilde{q}(\theta)) - \theta \tilde{q}(\theta) \Big) - (1-\alpha) \tilde{u}(\theta)& \\
&= \textsc{RS}_{\alpha}(\theta,\widetilde{M})& \\
&\ge \textsc{RS}_{\alpha}(\theta,M) 
\end{align*}
where the last inequality follows because $\widetilde{M}$ dominates $M$. 
{\hfill $\blacksquare{}$ \vspace{12pt}}

\begin{lemma}\label{lem:D_1}
If $D_{\tilde{q}\tilde{r}}$ is non-empty, then $u(\theta) > \tilde{u}(\theta)$ for every $\theta \in D_{\tilde{q}\tilde{r}}$. 
\end{lemma}
\begin{proof}
Assume for contradiction that there exists $\theta \in D_{\tilde{q}\tilde{r}}$ such that $u(\theta) \le \tilde{u}(\theta)$. 
Since $\theta \in D_{\tilde{q}\tilde{r}}$, we have $\theta > \underline{\theta}$.
Let $D:=\{y \in (\underline{\theta},\theta): y \notin D_{\tilde{q}\tilde{r}}\}$. We first show that $D$ is non-empty. 
If $D$ is empty, for all $y \in (\underline{\theta},\theta)$, we have $y \in D_{\tilde{q}\tilde{r}}$, i.e., $q(y)r(y) < \tilde{q}(y)\tilde{r}(y)$. But then 
\begin{align*}
u(\underline{\theta}) = \int \limits_{\underline{\theta}}^{\theta}q(y)r(y)dy + u(\theta)
< \int \limits_{\underline{\theta}}^{\theta}\tilde{q}(y)\tilde{r}(y)dy + u(\theta)
\le \int \limits_{\underline{\theta}}^{\theta}\tilde{q}(y)\tilde{r}(y)dy + \tilde{u}(\theta)
= \tilde{u}(\underline{\theta})
\end{align*}
where the second inequality follows from our assumption that $u(\theta) \le \tilde{u}(\theta)$. Observe $u(\underline{\theta}) < \tilde{u}(\underline{\theta})$ contradicts inequality $\textsc{RS}_{\alpha}(\underline{\theta},\widetilde{M}) \ge \textsc{RS}_{\alpha}(\underline{\theta},M)$. Therefore, $D$ is non-empty. Define $\theta_h$ as follows: 
\begin{align*} 
\theta_h &:= \sup \{y \in D\}.
\end{align*}
Since each $y \in D$ satisfies $q(y)r(y) \ge \tilde{q}(y)\tilde{r}(y)$, by left-continuity of $qr$ and $\tilde{q}\tilde{r}$, we get $q(\theta_h)r(\theta_h) \ge \tilde{q}(\theta_h)\tilde{r}(\theta_h)$, i.e., $\theta_h \notin D_{\tilde{q}\tilde{r}}$. This implies that $\theta_h < \theta$ (since $\theta \in D_{\tilde{q}\tilde{r}}$). Hence, we have $y \in D_{\tilde{q}\tilde{r}}$ for all $y \in (\theta_h,\theta]$, and this implies
\begin{equation}\label{eq:rent_order_1}
\tilde{u}(\theta_h)  = \int \limits_{\theta_h}^{\theta}\tilde{q}(y)\tilde{r}(y)dy + \tilde{u}(\theta) > \int\limits_{\theta_h}^{\theta}q(y)r(y)dy + \tilde{u}(\theta) \ge  \int\limits_{\theta_h}^{\theta}q(y)r(y)dy + u(\theta) = u(\theta_h).
\end{equation}
We now complete the proof by comparing the regulator's surplus under $M$ and $\widetilde{M}$ for $\theta_h$.
\begin{align*}
 \textsc{RS}_{\alpha}(\theta_h,\widetilde{M}) &=\quad 
 \tilde{r}(\theta_h)\Big(\overline{V}(\tilde{q}(\theta_h)) - \theta_h \tilde{q}(\theta_h) \Big)- (1-\alpha) \tilde{u}(\theta_h)& \\
 &\le_{(a)} r(\theta_h)\Big(\overline{V}(q(\theta_h)) - \theta_h q(\theta_h) \Big) - (1-\alpha) \tilde{u}(\theta_h)& \\
 &<_{(b)} r(\theta_h)\Big(\overline{V}(q(\theta_h)) - \theta_h q(\theta_h) \Big) - (1-\alpha) u(\theta_h)& \\
 &=\quad \textsc{RS}_{\alpha}(\theta_h, M). &
\end{align*}
To understand inequality $(a)$, recall $\tilde{q}(\theta_h)\tilde{r}(\theta_h)\le q(\theta_h)r(\theta_h)$. If this inequality holds with equality, we have $\tilde{q}(\theta_h)=q(\theta_h)$ and $\tilde{r}(\theta_h)=r(\theta_h)$ (as $M$ and $\widetilde{M}$ are both floor-randomized mechanisms). Otherwise, we have $\tilde{q}(\theta_h) \le q(\theta_h)$ and $\tilde{r}(\theta_h) \le r(\theta_h)$ with one of the inequalities being strict. Inequality $(a)$ follows from combining these inequalities with $q(\theta_h) \le q_e(\theta_h)$, which holds because of $M$ satisfies DD. The strict inequality $(b)$ follows from inequality (\ref{eq:rent_order_1}). The inequality $\textsc{RS}_{\alpha}(\theta_h,\widetilde{M}) < \textsc{RS}_{\alpha}(\theta_h, M)$ contradicts to the fact that $\widetilde{M}$ dominates $M$. Thus, $u(\theta) > \tilde{u}(\theta)$ for every $\theta \in D_{\tilde{q}\tilde{r}}$.
\end{proof}

\section{Proofs of Section \ref{sec:implications}}
\label{app:implications}
The proof of Proposition  \ref{prop:maxmin} is given after the proof of Lemma \ref{lem:mono}.\\

\noindent {\sc\textbf{Proof of Lemma \ref{lem:mono}.}} Let $M \equiv (r,q,u)$ be a floor-randomized mechanism satisfying DD. It suffices to prove that for all $\theta, \theta'$ such that $\theta' > \theta$,
\begin{align*}
\textsc{RS}_{\alpha}(\theta,M) \ge  \textsc{RS}_{\alpha}(\theta',M).
\end{align*}
Let $\theta^*_q$ be such that $r(\theta) =0$ for all $\theta > \theta^*_q$. We first establish $\textsc{RS}_{\alpha}(\theta^*_q,M) \ge \textsc{RS}_{\alpha}(\theta,M) \ge  \textsc{RS}_{\alpha}(\theta',M) \ge 0$, where $\theta^*_q < \theta < \theta'$. Clearly, 
$\textsc{RS}_{\alpha}(\theta,M)= \textsc{RS}_{\alpha}(\theta',M)=0$, and thus, it suffices to show that the first inequality holds. If $r(\theta^*_q)=0$, we have $\textsc{RS}_{\alpha}(\theta^*_q,M)=0$. If $r(\theta^*_q) > 0$, by Observation \ref{ob:hatq_bounds}, $\textsc{TS}(\theta^*_q,q(\theta^*_q)) > 0$. Since $u(\theta^*_q)=0$, we get $\textsc{RS}_{\alpha}(\theta^*_q,M) > 0$.

Henceforth, we assume that $\theta < \theta' \le \theta^*_q$ and show that $\textsc{RS}_{\alpha}(\theta,M) \ge  \textsc{RS}_{\alpha}(\theta',M)$ with a strict inequality holding if $\alpha \in (0,1)$. Both $q$ and $r$ are decreasing functions because $M$ is floor randomized. Further, by Observation \ref{ob:hatq_bounds}, total surplus at $\theta$ and $\theta'$ are non-negative.
Then,
    \begin{align} 
    &\textsc{RS}_{\alpha}(\theta,M) - \textsc{RS}_{\alpha}(\theta',M)& \nonumber\\
    &= \Big( \overline{V}(q(\theta)) - \theta q(\theta) \Big)r(\theta) - \Big( \overline{V}(q(\theta')) - \theta'q(\theta') \Big) r(\theta') -(1-\alpha)\Big(u(\theta)-u(\theta')\Big) &\nonumber\\
    &\ge \Big( \overline{V}(q(\theta)) - \theta q(\theta) \Big)r(\theta) - \Big( \overline{V}(q(\theta')) - \theta'q(\theta') \Big) r(\theta) -(1-\alpha)\Big(u(\theta)-u(\theta')\Big) & \nonumber \\
    & \textrm{(since $r(\theta) \ge r(\theta')$ and $\textsc{TS}(\theta',q(\theta')) \ge 0$)} \nonumber\\
    &= r(\theta) \int \limits_{q(\theta')}^{q(\theta)}(P(z)-\theta)dz + r(\theta) (\theta'-\theta)q(\theta') -(1-\alpha)\Big(\int \limits_{\theta}^{\theta'}q(\tilde{\theta})r(\tilde{\theta})d\tilde{\theta}\Big)~\qquad~\textrm{(by IC)} \nonumber\\
    &\ge r(\theta) \int \limits_{q(\theta')}^{q(\theta)}(P(z)-\theta)dz + r(\theta) (\theta'-\theta)q(\theta') -(1-\alpha)r(\theta)\Big(\int \limits_{\theta}^{\theta'}q(\tilde{\theta})d\tilde{\theta}\Big)~\qquad~\textrm{(since $r$ is decreasing)} \nonumber\\
    &= r(\theta)\Bigg[\int \limits_{q(\theta')}^{q(\theta)}(P(z)-\theta)dz + (\theta'-\theta)q(\theta') -\int \limits_{\theta}^{\theta'}q(\tilde{\theta})d\tilde{\theta}+ \alpha \int \limits_{\theta}^{\theta'}q(\tilde{\theta})d\tilde{\theta}\Bigg] &\nonumber\\
    &= r(\theta)\Bigg[\int \limits_{q(\theta')}^{q(\theta)}(P(z)-\theta)dz -\int \limits_{\theta}^{\theta'}(q(\tilde{\theta})-q(\theta'))d\tilde{\theta} + \alpha \int \limits_{\theta}^{\theta'}q(\tilde{\theta})d\tilde{\theta}\Bigg].  & \label{eq:net}
    \end{align}
Observe that the third integral is always non-negative and it is positive if $\alpha \in (0,1)$. Thus, it suffices to show that the first integral is weakly larger than the second. Observe that the first integral can be equivalently written in terms of $y \equiv P(z)$ instead of $z$. Formally,
    \begin{align}
        \int \limits_{q(\theta')}^{q(\theta)}(P(z)-\theta)dz &= \int \limits_{\theta}^{P(q(\theta))}(q(\theta)-q(\theta'))dy + \int \limits_{P(q(\theta))}^{P(q(\theta'))}(P^{-1}(y)-q(\theta'))dy. \label{eq:change_var}
    \end{align}
    Figure \ref{fig_lemma_2} pictorially depicts this change of variable. Observe $P(q(\theta)) \ge \theta$ by DD and $P(q(\theta')) \ge P(q(\theta))$ because $P$ and $q$ are decreasing. Notice that each integral on the right-hand side of (\ref{eq:change_var}) is non-negative because $P$ and $q$ are decreasing.
     
    \begin{figure}[!hbt]
        \centering
        \includegraphics[width=4in]{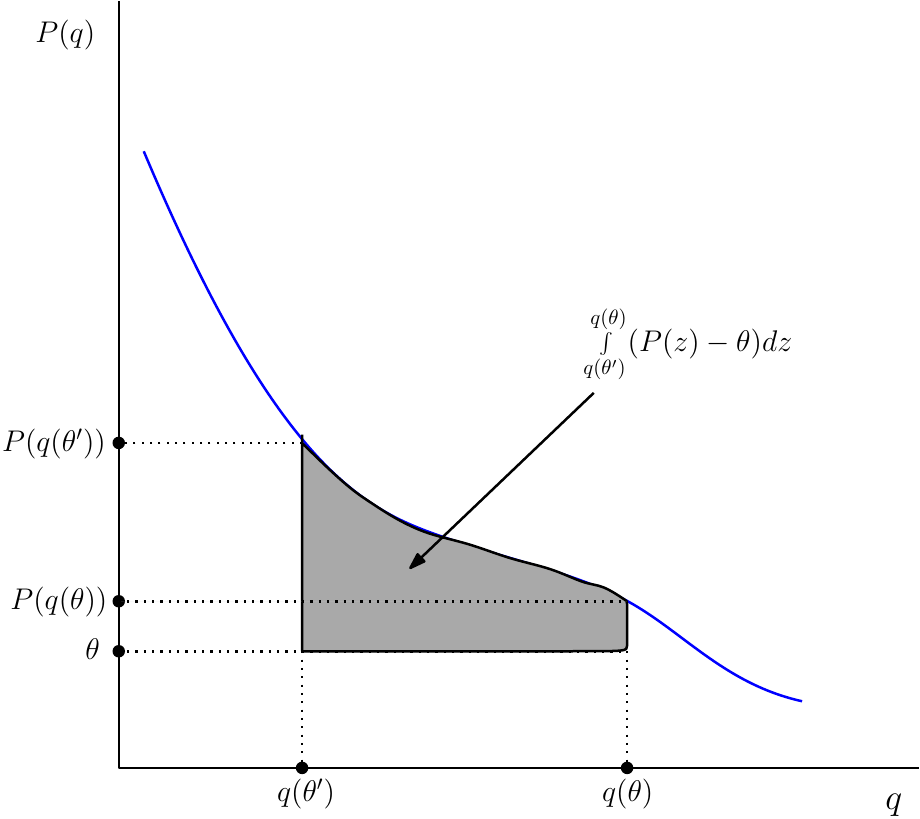}
        \caption{Pictorial depiction of integration along the price axis}
        \label{fig_lemma_2}
    \end{figure}
    
    Now, suppose $\theta' \le P(q(\theta))$, then (\ref{eq:change_var}) implies 
    \begin{align}
        \int \limits_{q(\theta')}^{q(\theta)}(P(z)-\theta)dz &\ge \int \limits_{\theta}^{\theta'}(q(\theta)-q(\theta'))d\tilde{\theta} 
        \ge \int \limits_{\theta}^{\theta'}(q(\tilde{\theta})-q(\theta'))d\tilde{\theta} \label{eq:er1}
    \end{align}
    where the second inequality holds because $q$ is decreasing. If $\theta' > P(q(\theta))$, we know by DD, $q_e(\theta') \ge q(\theta')$, and hence, $P(q(\theta')) \ge P(q_e(\theta')) = \theta'$.
    This implies that $\theta' \in [P(q(\theta)),P(q(\theta'))]$. Again, using equation (\ref{eq:change_var}), we get 
    \begin{align}
        \int \limits_{q(\theta')}^{q(\theta)}(P(z)-\theta)dz &\ge \int \limits_{\theta}^{P(q(\theta))}(q(\theta)-q(\theta'))d\tilde{\theta} + \int \limits_{P(q(\theta))}^{\theta'}(P^{-1}(\tilde{\theta})-q(\theta'))d\tilde{\theta} &\nonumber \\
        &\ge \int \limits_{\theta}^{P(q(\theta))}(q(\theta)-q(\theta'))d\tilde{\theta} + \int \limits_{P(q(\theta))}^{\theta'}(q(\tilde{\theta})-q(\theta'))d\tilde{\theta}&\textrm{(by DD)}\nonumber \\
        &\ge \int \limits_{\theta}^{P(q(\theta))}(q(\tilde{\theta})-q(\theta'))d\tilde{\theta} + \int \limits_{P(q(\theta))}^{\theta'}(q(\tilde{\theta})-q(\theta'))d\tilde{\theta} &\textrm{(since $q$ is decreasing)}\nonumber \\
        &= \int \limits_{\theta}^{\theta'}(q(\tilde{\theta})-q(\theta'))d\tilde{\theta}.& \label{eq:er2}
    \end{align}
By inequalities (\ref{eq:er1}) and (\ref{eq:er2}), we have
    \begin{align*}
        \int \limits_{q(\theta')}^{q(\theta)}(P(z)-\theta)dz &\ge \int \limits_{\theta}^{\theta'}(q(\tilde{\theta})-q(\theta'))d\tilde{\theta}. 
    \end{align*} 
Thus, (\ref{eq:net}) implies that $\textsc{RS}_{\alpha}(\theta,M) \ge \textsc{RS}_{\alpha}(\theta',M)$ with a strict inequality holding for $\alpha \in (0,1)$.{\hfill $\blacksquare{}$ \vspace{12pt}}

\noindent {\sc \textbf{Proof of Proposition \ref{prop:maxmin}.}} We begin by arguing that it is without loss to focus attention on floor-randomized mechanisms satisfying DD and left-continuity. To see why, pick any IC and IR mechanism $M$ that violates either of these properties. By Theorem \ref{theo:undnec}, there exists a floor-randomized mechanism satisfying DD and left-continuity, say $M'$, that dominates $M$. Moreover, 
\begin{align*}
    \inf_{G \in \mathcal{G}} \textsc{RS}_{\alpha}(M',G) =_{(a)} \textsc{RS}_{\alpha}(M',G^\star) &\ge_{(b)} \textsc{RS}_{\alpha}(M,G^\star) \ge \inf_{G \in \mathcal{G}} \textsc{RS}_{\alpha}(M,G)
\end{align*}
where equality $(a)$ holds because of Lemma \ref{lem:mono}. As for inequality $(b)$, it holds because $M'$ dominates $M$. Notice that the set of floor-randomized mechanisms satisfying DD and left-continuity is non-empty, and therefore, there exists a max-min optimal mechanism which satisfies these three properties.

To complete the proof, it suffices to argue that $M_{G^\star}$, where
\begin{align}\label{eq:mm2}
 M_{G^\star} \in \arg \max_{M \in \mathcal{N}}~ \textsc{RS}_\alpha(M,G^\star)   
\end{align}
is indeed max-min optimal in the set $\mathcal{N}$ of floor-randomized mechanisms satisfying DD and left-continuity. By Lemma \ref{lem:mono}, infimum occurs at  $G^\star$ for all such mechanisms. And because of equation (\ref{eq:mm2}), mechanism $M_{G^\star}$ attains the maximum regulator's surplus among all such mechanisms at $G^\star$. {\hfill $\blacksquare{}$ \vspace{12pt}}

\noindent {\sc \textbf{Proof of Proposition \ref{prop:max}}.} Let  $L^{\infty}(\Theta)$ be the set of bounded density functions on $\Theta$. To prove the result, it suffices to show that for a given undominated mechanism $M$ satisfying the conditions of the proposition, there exists density $f \in L^{\infty}(\Theta)$ such that
\[
\textsc{RS}_\alpha(M,f) := \int \limits_{\underline{\theta}}^{\overline{\theta}}\textsc{RS}_{\alpha}(\theta, M)f(\theta)d\theta \ge \textsc{RS}_\alpha(M',f) := \int \limits_{\underline{\theta}}^{\overline{\theta}}\textsc{RS}_{\alpha}(\theta, M')f(\theta)d\theta
\]
for every IC and IR mechanism $M'$.\footnote{Earlier in the paper, the expected regulator surplus for a mechanism $M$ was denoted by $\textsc{RS}_{\alpha}(M,F)$, where $F$ was a cumulative distribution function.  In this proof, to make the analysis transparent, we use the notation $\textsc{RS}_\alpha(M,f)$, where $f$ is a density.}

Observe, $L^{\infty}(\Theta)$ is weak* compact.
For each $f \in L^{\infty}(\Theta)$, select $M_f \in \mathcal{M}$, an IC and IR mechanism, such that $\textsc{RS}_\alpha(M_f,f) > \textsc{RS}_\alpha(M,f)$. If no such $M_f$ exists for some $f\in L^{\infty}(\Theta)$, then the proof is complete.

By continuity, there is a weak* open neighborhood $O_f$ of $f$ such that for every $f' \in O_f$
\begin{align*}
    \textsc{RS}_\alpha(M_f,f') > \textsc{RS}_\alpha(M,f').
\end{align*}

The collection $\{O_f: f \in L^{\infty}(\Theta)\}$ is an open cover of $L^{\infty}(\Theta)$. Compactness implies that there exists a finite subcover $\{O_{f_n}:n=1,2,...,N\}$. For each $n \in \{1,\ldots,N\}$, let $O_{f_n}$ be the open neighborhood corresponding to some density $f_n$ and $M_n$ be the IC and IR mechanism such that 
for every $f' \in O_{f_n}$
\begin{align*}
    \textsc{RS}_\alpha(M_n,f') > \textsc{RS}_\alpha(M,f').
\end{align*}
Define the following set of vectors in $\Re^N$:
\begin{align*}
    \mathcal{H}=\{ \big(\textsc{RS}_\alpha(M_1,f) - \textsc{RS}_\alpha(M,f), \ldots, \textsc{RS}_\alpha(M_n,f) - \textsc{RS}_\alpha(M,f)\big): f \in L^{\infty}(\Theta)\}.
\end{align*}
Note that for any $h \in \mathcal{H}$, $h_i > 0$ for some $i \in \{1,\ldots,N\}$. As a result, $\mathcal{H} \cap \Re_-$ is empty. Since $\mathcal{H}$ is convex, there is a separating hyperplane $\gamma \in \Re^N_+$ separating $\mathcal{H}$ and $\Re_-$ such that $\gamma \cdot h > 0$ for all $h \in \mathcal{H}$. Without loss of generality, we can normalize $\gamma$ such that $\sum_i^N \gamma_i =1$. Hence, for every $f$, we have 
\begin{align}
    \label{eq:ra1}
    \sum_{i=1}^N \gamma_i \textsc{RS}_\alpha(M_i,f) > \textsc{RS}_\alpha(M,f).
\end{align}
Lemma \ref{lem:floor} (which generalizes Theorem \ref{theo:fr} and is proved after this proof) establishes that there exists a floor-randomized mechanism $M^T$ such that for every $f$,
\begin{align}
    \label{eq:cont1}
    \textsc{RS}_\alpha(M^T,f) \ge \sum_{i=1}^N \gamma_i \textsc{RS}_\alpha(M_i,f).
\end{align}
Thus, combining (\ref{eq:ra1}) and (\ref{eq:cont1}), we get that there exists a floor-randomized mechanism $M^T$ such that for every $f$,
\begin{align}
    \label{eq:cont2}
    \textsc{RS}_\alpha(M^T,f) > \textsc{RS}_\alpha(M,f).
\end{align}
Further, we can assume $M^T$ to be left continuous and satisfies DD (by Theorem \ref{theo:undnec}). By (\ref{eq:cont2}), we can conclude that for almost all $\theta \in \Theta$,
\begin{align}
    \label{eq:cont3}
    \textsc{RS}_\alpha(\theta, M^T) \ge \textsc{RS}_\alpha(\theta,M)
\end{align}
with strict inequality for a positive Lebesgue measure of $\theta$s. By Lemma \ref{lem:newl1}, both $\textsc{RS}_\alpha(\theta, M^T)$ and $\textsc{RS}_\alpha(\theta,M)$ are left continuous for all $\theta > \underline{\theta}$. Hence, inequality (\ref{eq:cont3}) holds for all $\theta > \underline{\theta}$. Suppose $\textsc{RS}_\alpha(\underline{\theta},M) > \textsc{RS}_\alpha(\underline{\theta}, M^T)$. By continuity of $\textsc{RS}_\alpha(\theta,M)$ at $\theta=\underline{\theta}$ and the fact $\textsc{RS}_\alpha(\underline{\theta}, M^T)$ is decreasing (Lemma \ref{lem:mono}), we conclude that there exists $\theta' > \underline{\theta}$ such that $\textsc{RS}_\alpha(\theta,M) > \textsc{RS}_\alpha(\theta, M^T)$ for all $\theta \in [\underline{\theta},\theta']$, contradicting our earlier conclusion. Hence, inequality (\ref{eq:cont3}) also holds for $\theta=\underline{\theta}$.

Thus, we have shown that inequality (\ref{eq:cont3}) holds for all $\theta$ and is strict for a positive Lebesgue measure of $\theta$s. This contradicts the fact that $M$ is undominated. {\hfill $\blacksquare{}$ \vspace{12pt}}

The following lemma generalizes Theorem \ref{theo:fr}.\footnote{By taking $N=1$ in Lemma \ref{lem:floor}, we recover Theorem \ref{theo:fr}.}
\begin{lemma}
    \label{lem:floor}
    Suppose $(M_1,\ldots,M_N)$ are $N$ IC and IR mechanisms. Then, there exists a floor-randomized mechanism $M^T$ such that for any $\gamma \in [0,1]^N$ with $\sum_{i=1}^N \gamma_i =1$ we have 
    \begin{align*}
        \textsc{RS}_\alpha(\theta, M^T) &\ge \sum_{i=1}^N \gamma_i \textsc{RS}_\alpha(\theta,M_i)~\qquad~\forall~\theta.
    \end{align*}
\end{lemma}
\begin{proof}
    Let $\tilde{q}(\theta)\equiv \sum \limits_{i=1}^{N}\gamma_iq_i(\theta)r_i(\theta)$, which decreases in $\theta$ because each $M_i$ is IC, and define
\begin{align*}
    q^T(\theta) = \begin{cases}\max(\tilde{q}(\theta),\hat{q}) & \textrm{$\tilde{q}(\theta)>0$} \\
        0 & {\rm otherwise}
    \end{cases} & &  r^T(\theta) = \begin{cases}\min(\frac{\tilde{q}(\theta)}{\hat{q}},1)& \textrm{$\tilde{q}(\theta)>0$} \\
    0 & {\rm otherwise}
    \end{cases}
\end{align*}
At every $\theta$, we have $q^T(\theta)r^T(\theta)$ equals $\tilde{q}(\theta)$. Consequently, the mechanism $M^T\equiv(r^T,q^T,u^T)$, where $u^T(\theta):= \int_{\theta}^{\overline{\theta}}\tilde{q}(z)dz$ for all $\theta$, is IC and IR with $u^T(\overline{\theta})=0$. Moreover, it is also floor randomized with $\Theta_{1} = \{\theta: \tilde{q}(\theta) \ge \hat{q}\}$ and $\Theta_{01} = \{\theta: \hat{q} > \tilde{q}(\theta) > 0\}$. Notice that $r^T(\theta) \in (0,1)$ is true only when $q^T(\theta)=\hat{q}$. 

To establish the desired inequality, pick $\theta \in \Theta$. If $\tilde{q}(\theta)=0$, then $q_i(\theta)=0$ for all $i$ and the claim holds trivially. So, assume $\tilde{q}(\theta) > 0$ and consider the following two cases.

\noindent \textbf{\textsc{Case 1}}. Let $\theta$ be such that $\tilde{q}(\theta)\ge \hat{q}$. Here, by definition, $r^T(\theta)=1$ and $q^T(\theta)=\tilde{q}(\theta)$.
\begin{align*}
   \textsc{RS}_{\alpha}(\theta,M^T) &=   r^T(\theta)\overline{V}(q^T(\theta)) - \theta q^T(\theta)r^T(\theta) - (1-\alpha)u^T(\theta) \\
   &= \overline{V}(\tilde{q}(\theta))-\theta \tilde{q}(\theta)- (1-\alpha)\Big(\sum_{i=1}^{N}\gamma_i u_i(\theta)-\sum_{i=1}^{N}\gamma_i u_i(\overline{\theta})\Big)\\
   &\ge_{(a)} \overline{V}(\tilde{q}(\theta))-\theta \tilde{q}(\theta)- (1-\alpha)\sum_{i=1}^{N}\gamma_i u_i(\theta)\\
    &= co(\widetilde{V})(\tilde{q}(\theta))- \theta \tilde{q}(\theta)-(1-\alpha)\sum_{i=1}^{N}\gamma_i u_i(\theta)\\
    &\ge_{(b)} \sum_{i=1}^{N}\gamma_i co(\widetilde{V})(q_i(\theta)r_i(\theta))- \theta \tilde{q}(\theta)- (1-\alpha)\sum_{i=1}^{N}\gamma_i u_i(\theta)\\
    &\ge_{(c)} \sum_{i=1}^{N}\gamma_i r_i(\theta)\overline{V}(q_i(\theta))- \theta \tilde{q}(\theta)- (1-\alpha)\sum_{i=1}^{N}\gamma_i u_i(\theta) \\
    &= \sum_{i=1}^N \gamma_i \textsc{RS}_\alpha(\theta,M_i),
\end{align*}
where inequality $(a)$ holds because each $M_i$ is IR, and therefore, $u_i(\overline{\theta}) \ge0$. As for inequality $(b)$, it follows from concavity of $co(\widetilde{V})$. Finally, inequality $(c)$ holds because $co(\widetilde{V})$ is a concave closure of function $\widetilde{V}(\theta) = \overline{V}(\theta)\mathbb{I}_{\{q>0\}}$.

\noindent \textbf{\textsc{Case 2}}. Let $\theta$ be such that $0< \tilde{q}(\theta)< \hat{q}$. Then, $q^T(\theta)=\hat{q}$ and $r^T(\theta)=\tilde{q}(\theta)/\hat{q}$. Moreover,
\begin{align*}
   \textsc{RS}_{\alpha}(\theta,M^T)&=   r^T(\theta)\overline{V}(q^T(\theta)) - \theta q^T(\theta)r^T(\theta) - (1-\alpha)u^T(\theta) \\
   &= r^T(\theta)\overline{V}(\hat{q})-\theta \tilde{q}(\theta)- (1-\alpha)\Big(\sum_{i=1}^{N}\gamma_i u_i(\theta)-\sum_{i=1}^{N}\gamma_i u_i(\overline{\theta})\Big)\\
   &\ge_{(a)} r^T(\theta)\overline{V}(\hat{q})-\theta \tilde{q}(\theta)- (1-\alpha)\sum_{i=1}^{N}\gamma_i u_i(\theta)\\
    &= co(\widetilde{V})(\tilde{q}(\theta))- \theta \tilde{q}(\theta)-(1-\alpha)\sum_{i=1}^{N}\gamma_i u_i(\theta)\\
    &\ge_{(b)} \sum_{i=1}^{N}\gamma_i co(\widetilde{V})(q_i(\theta)r_i(\theta))- \theta \tilde{q}(\theta)- (1-\alpha)\sum_{i=1}^{N}\gamma_i u_i(\theta)\\
    &\ge_{(c)} \sum_{i=1}^{N}\gamma_i r_i(\theta)\overline{V}(q_i(\theta))- \theta \tilde{q}(\theta)- (1-\alpha)\sum_{i=1}^{N}\gamma_i u_i(\theta) \\
    &= \sum_{i=1}^N \gamma_i \textsc{RS}_\alpha(\theta,M_i),
\end{align*}
where inequality $(a)$ holds because each $M_i$ is IR, and therefore, $u_i(\overline{\theta}) \ge0$. As for inequality $(b)$, it follows from concavity of $co(\widetilde{V})$. Finally, inequality $(c)$ holds because $co(\widetilde{V})$ is a concave closure of function $\widetilde{V}(\theta) = \overline{V}(\theta)\mathbb{I}_{\{q>0\}}$.
\end{proof}

\begin{lemma}
    \label{lem:newl1}
    If $M$ is an IC and IR mechanism that is left continuous, floor randomized, and satisfies DD, then $\textsc{RS}_{\alpha}(\theta,M)$ is left continuous in $\theta$ for all $\theta > \underline{\theta}$.
\end{lemma}
\begin{proof}
    To see why $\textsc{RS}_{\alpha}(\theta,M)$ is left continuous in $\theta$,
    note that 
    \begin{align*}
        \textsc{RS}_{\alpha}(\theta,M) &= r(\theta) \textsc{TS}(\theta,q(\theta)) - (1-\alpha)u(\theta).
    \end{align*}
    Since $u$ is convex, it is continuous in the interior of $\Theta$. Also, since $u$ is decreasing, non-negative, and $u(\bar{\theta})=0$, it is continuous at $\bar{\theta}$. This implies that to show left-continuity of $\textsc{RS}_{\alpha}(\theta,M)$, it is sufficient to show that $r(\theta)\textsc{TS}(\theta,q(\theta))$ is left continuous at $\theta > \underline{\theta}$.
    
    Consider the partition of type space $\Theta$, $\{\Theta_1,\Theta_{01},\Theta_0\}$, defined for the floor-randomized mechanism $M$.
    \begin{enumerate}
        \item Consider the set $\Theta_1$, and let $\sup \Theta_1 = \theta_\ell$. If $\theta_{\ell} = \underline{\theta}$, there is nothing to prove. Therefore, assume $\theta_{\ell} > \underline{\theta}$. For every $\theta \in \Theta_1$, $r(\theta)=1$ and $q(\theta) \ge \hat{q}$. By left-continuity of $M$, we must have $q$ to be left continuous and $r(\theta_{\ell})=1$. 
        Consequently, $\theta_{\ell} \in \Theta_1$, and $r(\theta)=1$ for all $\theta \in \Theta_1$ and $q(\theta)$ is left continuous on $\Theta_1$. Hence, $r(\theta)\textsc{TS}(\theta,q(\theta))$ is left continuous for all $\theta \in (\underline{\theta},\theta_{\ell}]$.

        \item Now consider $\Theta_{01}$ (if it is empty, there is nothing to prove), and let $\sup \Theta_{01} = \theta_h$. For every $\theta \in \Theta_{01}$, we have $q(\theta)=\hat{q}$. If $\theta_h \in \Theta_{01}$, then $r(\theta) \in (0,1)$ for all $\theta \in \Theta_{01}$, and since $q(\theta) = \hat{q}$ for all $\theta \in \Theta_{01}$, left-continuity of $M$ ensures left-continuity of $r$. If $\theta_h \notin \Theta_{01}$, then $r(\theta_h)=0$ and $q(\theta_h)=0$. Left-continuity of $M$ ensures $\lim_{\theta \uparrow \theta_h}r(\theta)=0$. As a result, $r(\theta)\textsc{TS}(\theta,q(\theta))$ is left continuous for all $\theta \in (\theta_{\ell},\theta_h]$.

        \item Finally, consider any $\theta > \theta_h$. Since $q(\theta) = r(\theta) = 0$ for $\theta > \theta_h$, $r(\theta)\textsc{TS}(\theta, q(\theta))$ is left continuous for all $\theta \in (\theta_h,\bar{\theta}]$.
    \end{enumerate}
    This concludes the proof that $\textsc{RS}_{\alpha}(\theta,M)$ is left continuous in $\theta$ for all $\theta > \underline{\theta}$.
\end{proof}

\noindent {\sc\textbf{Proof of Lemma \ref{lem:extreme}.}} Let $(r,q,u) \in \mathcal{M}(q)$ be such that $r(\theta) \in (0,1)$ for some $\theta$. It suffices to show that $(r,q,u)$ is not an extreme point of $\mathcal{M}(q)$. That is, it is enough to establish that there exists two decreasing functions $\hat{r}$ and $\tilde{r}$, both taking values in $[0,1]$, and satisfying
\[r(\theta) = \frac{1}{2}\hat{r}(\theta) + \frac{1}{2}\tilde{r}(\theta) ~~\forall~\theta.\]
Notice that the above equation, in combination with $q,$ will then imply $u = \left(\hat{u}(\theta) + \tilde{u}(\theta)\right)/2$, establishing that $(r,q,u)$ can be written as convex combination of two mechanism, $(\hat{r},q,\hat{u})$ and $(\tilde{r},q,\tilde{u})$, both in $\mathcal{M}(q)$.

To do so, consider 
\begin{align*} 
\hat{r}(\theta) = 
\begin{cases}
0 & \textrm{if}~r(\theta) < 0.5 \\
2r(\theta) - 1 & \textrm{if}~r(\theta) \ge 0.5
\end{cases}
\hspace{0.5in}
\tilde{r}(\theta) = 
\begin{cases}
2r(\theta) & \textrm{if}~r(\theta) < 0.5 \\
1 & \textrm{if}~r(\theta) \ge 0.5
\end{cases}
\end{align*}
It is easy to see that $r(\theta)= (\hat{r}(\theta)+\tilde{r}(\theta))/2$ for all $\theta$.
Both $\hat{r}$ and $\tilde{r}$ are feasible, that is, they take values in $[0,1]$. Finally, both $\hat{r}$ and $\tilde{r}$ are decreasing because $r$ is decreasing. If $r(\theta) \notin \{0,1\}$ for some $\theta$, then $\hat{r}(\theta) \ne \tilde{r}(\theta)$. Therefore, $r$ can be written as convex combination of two more decreasing functions. {\hfill $\blacksquare{}$ \vspace{12pt}}

\noindent {\sc \textbf{Proof of Proposition \ref{prop:pshrink}}.} We begin by showing that the new quantity floor $\hat{q}_{{\rm n}} > \hat{q}$, and use this inequality to establish that if $M \in \mathcal{D}(P_{{\rm n}})$, then $M \in \mathcal{D}(P)$. Observe, for any $q$,
    \begin{align*}
        V_{{\rm n}}(q) - q P_{{\rm n}}(q) &= \int \limits_0^{q}\big[P_{{\rm n}}(z) - P_{{\rm n}}(q)\big]dz \\
        &= \int \limits_0^{q} \Big[ \big[P(q) - P_{{\rm n}}(q)\big] - \big[P(z) - P_{{\rm n}}(z)\big] \Big]dz + \int \limits_0^{q} \big[P(z) - P(q)\big]dz \\
        &< \int \limits_0^{q} \big[P(z) - P(q)\big]dz \\
        &= V(q) - qP(q)
    \end{align*}
    where the inequality follows from (\ref{eq:decdiff}). Further, applying the above inequality for $q=\hat{q}$, we get 
    \begin{align*}
        V_{{\rm n}}(\hat{q}) - \hat{q} P_{{\rm n}}(\hat{q}) < c.
    \end{align*}
The above inequality along with the fact that $V_{{\rm n}}(q) - q P_{{\rm n}}(q)$ is strictly increasing and continuous in $q$ implies that $\hat{q}_{{\rm n}} > \hat{q}$. 

    Now pick any deterministic undominated mechanism $(q,u) \in \mathcal{D}(P_{{\rm n}})$ -- note that for deterministic mechanisms $q$ itself determines $r$ and $u$ is determined by the envelope formula. We now establish $(q,u) \in \mathcal{D}(P)$. Observe, $(q,u)$ is floor-randomized with respect to the floor $\hat{q}$. This is because, by Theorem \ref{theo:undnec}, if $q(\theta)>0$, then $q(\theta) \ge \hat{q}_{{\rm n}} > \hat{q}$. Hence, $(q,u)$ is floor-randomized with respect to $P$. Further, by Theorem \ref{theo:undnec}, $(q,u)$ satisfies DD with respect to $P_{{\rm n}}$, which means for any $\theta$, we have $q(\theta) \le P^{-1}_{{\rm n}}(\theta) \le  P^{-1}(\theta)$  with second inequality strict for $\theta > \underline{\theta}$ and all inequalities equalities for $\theta = \underline{\theta}$ (because $P$ and $P_{{\rm n}}$ cross at $P^{-1}(\underline{\theta})$ and (\ref{eq:decdiff}) holds). As a result $(q,u)$ satisfies strict DD with respect to $P$. Since $(q,u) \in  \mathcal{D}(P_{{\rm n}})$, it is left continuous. Thus, by Theorem \ref{theo:undsuff}, we have $(q,u) \in \mathcal{D}(P)$. {\hfill $\blacksquare{}$ \vspace{12pt}}

\section{Unknown fixed cost and known marginal cost}
\label{sec:fixed}
A key step in our analysis was a rent-invariant transformation, where we converted a given IC and IR mechanism into another IC and IR mechanism that had type-by-type the same rent ($u(\theta)$). Among such rent-invariant transformations, the floor-randomized transformation improved the expected total surplus ($r(\theta)\left( V(q(\theta))-c-\theta q(\theta)\right)$) at every type and allowed us to narrow the search for undominated mechanisms to a more structured class of IC and IR mechanisms. However, if the private information of the monopolist affects the fixed cost, a rent-invariant transformation may not always be available, and even when it is available, it may adversely impact the expected total surplus at some types. Consequently, the search for undominated mechanisms cannot be confined to a well-behaved smaller subset of mechanisms, further complicating the analysis. Nevertheless, below we make progress in characterizing undominated mechanisms in a special case of the bilinear cost of \citet{BM82}, where the fixed cost is the private information of the monopolist, and the marginal cost is commonly known.

Suppose the cost of producing a quantity $q>0$ of the product is $\theta + cq$, where $c > 0$ is the constant marginal cost that is observed by both the monopolist and the regulator. However, the fixed cost $\theta$ is privately observed by the monopolist.\footnote{The result in this section extends to the cost specification $k_1\theta + c q$, where $k_1 >0$ is such that $[A2]$ is satisfied.} Thus, the private information only affects the fixed cost and the marginal cost is common knowledge. In this case, $M \equiv (r,q,u)$ is IC and IR if and only if $r(\theta)$ is decreasing in $\theta$, and
\[
u(\theta) = u(\overline{\theta}) + \int \limits_{\theta}^{\overline{\theta}}r(y)dy
\]
where $u(\overline{\theta}) \ge 0$. Let $q_e$ be such that $q_e = P^{-1}(c)$.

With known marginal cost, a rent-invariant transformation is unavailable as $r$ pins down the rent uniquely up to a constant. Therefore, we must search over the entire set of IC and IR mechanisms. However, because the marginal cost is commonly known, total surplus is maximized uniquely at $q_e$ and the characterization of undominated mechanisms reduces to obtaining restriction on $r$, which are identified in the next result.

\begin{theorem}
\label{theo:fixedc}
An mechanism  $M\equiv(r,q,u)$ is a undominated if and only if $r$ is left continuous with $r(\underline{\theta})=1$, $u(\bar{\theta})=0$, and $q(\theta)=q_e$ for all $\theta$.    
\end{theorem}
\begin{proof}
For an IC and IR mechanism $M \equiv (r,q,u)$, the regulator's surplus is
\[
\textsc{RS}_\alpha(\theta, M) = r(\theta)\left[V(q(\theta))-\theta-cq(\theta)\right]-(1-\alpha)u(\theta).
\]
Clearly, the total surplus $V(q(\theta))-\theta-cq(\theta)$ is uniquely maximized at $q(\theta)=q_e$ for all $\theta$. Moreover, by [A2], which in this context says that the total surplus $V(q_e) - \overline{\theta} - cq_e > 0$ at the efficient quantity and highest fixed cost is positive, $V(q_e) - \theta - cq_e > 0$ for all $\theta$.
Therefore, $M$ being undominated implies $q(\theta)=q_e$ for all $\theta$. Next, observe that by arguments similar to those in the main paper, $r(\underline{\theta})=1$ and $r$ is left continuous, and $u(\bar{\theta})=0$.

Now suppose that an IC and IR mechanism $M \equiv (r,q,u)$ is such that $r$ is left continuous with $r(\underline{\theta})=1$, $u(\bar{\theta})=0$, and $q(\theta)=q_e$ for all $\theta$. Assume for contradiction that $M$ is dominated by another IC and IR mechanism $\widetilde{M} = (\tilde{r},\tilde{q},\tilde{u})$. It is without loss of generality to assume that $\tilde{r}$ is left continuous with $\tilde{r}(\underline{\theta})=1$, $\tilde{u}(\bar{\theta})=0$, and $\tilde{q}(\theta)=q_e$ for all $\theta$. Let
\begin{align*}
    D_{r} &:= \{\theta: \tilde{r}(\theta) < r(\theta)\}.
\end{align*}
By Lemma \ref{lem:Dr} (which is stated and proved after this proof), $D_{r}$ has a positive measure. Next, define 
\begin{align*}
    \theta_h &:= \sup~\{\theta: \theta \in D_{r}\}.
\end{align*}
$\theta_h$ is a finite real number because $D_{r}$ is non-empty and bounded. Clearly, $\theta_h > \underline{\theta}$. In what follows, the proof obtains a contradiction by showing that there exists $\theta$, in a left neighbourhood of $\theta_h$, at which the regulator gets a higher surplus under mechanism $M$. 

There exists $\theta_{\ell} < \theta_h$, along with a corresponding left-neighbourhood $\mathcal{N} : = [\theta_{\ell}, \theta_h)$ of $\theta_h$,  such that 
\begin{align*}
 r(\theta) >_{(a)} \tilde{r}(\theta)\geq 0 ~\qquad~\forall~\theta\in\mathcal{N}.
\end{align*}
Inequality $(a)$ follows from the definition of $\theta_h$ along with the fact that $r$ and $\tilde{r}$ are decreasing and left continuous in $\theta$.

Now, let
    \[
    \Delta := \sup_{y \in \mathcal{N}} \big[r(y)-\tilde{r}(y)\big]
    \]
which is positive by definition of $\mathcal{N}$. Thus, there exists $\theta \in \mathcal{N}$ such that $r(\theta) - \tilde{r}(\theta)  = \Delta - \epsilon$, where $\epsilon<\Delta\frac{\textsc{TS}(\theta_h, q_e)}{\textsc{TS}(\underline{\theta}, q_e)}$ and $\textsc{TS}(\theta, q_e) = V(q_e)-\theta-cq_e$. Consider the difference
\begin{align*}
      \textsc{RS}_{\alpha}(\theta,M)-\textsc{RS}_{\alpha}(\theta,\widetilde{M}) &=\Big(r(\theta)-\tilde{r}(\theta)\Big)\textsc{TS}(\theta, q_e)-(1-\alpha)\big(u(\theta)-\tilde{u}(\theta)\big).
      \end{align*}
      To complete the proof, it suffices to show that the right hand side of the above equation is positive, contradicting the fact that $\widetilde{M}$ dominates $M$. Observe
      \begin{align*}
          & \Big(r(\theta)-\tilde{r}(\theta)\Big)\textsc{TS}(\theta, q_e)-(1-\alpha)\big(u(\theta)-\tilde{u}(\theta)\big)  \\
          &= \Big(r(\theta)-\tilde{r}(\theta)\Big)\textsc{TS}(\theta, q_e)- (1-\alpha) \Big( \int \limits_{\theta}^{\theta_h} \big(r(y)-\tilde{r}(y)\big)dy\Big)  - (1-\alpha)  \Big(\int \limits_{\theta_h}^{\overline{\theta}} \big(r(y)-\tilde{r}(y)\big)dy\Big)\\
          &~\qquad~\textrm{(since $u(\overline{\theta})=\tilde{u}(\overline{\theta})=0$)} \\
          &\ge \Big(r(\theta)-\tilde{r}(\theta)\Big)\textsc{TS}(\theta, q_e)- (1-\alpha) \Big(\int \limits_{\theta}^{\theta_h} \big(r(y)-\tilde{r}(y)\big)dy\Big) \\
          & \textrm{(since $r(y) \le \tilde{r}(y)$ for all $y > \theta_h$)} \\
          &\ge \Big(r(\theta)-\tilde{r}(\theta)\Big)\textsc{TS}(\theta, q_e)- \Big(\int \limits_{\theta}^{\theta_h} \big(r(y)-\tilde{r}(y)\big)dy\Big)~\qquad~\textrm{(since $r(y) > \tilde{r}(y)$ for all $y \in [\theta,\theta_h)$)} \\
          &\ge \Big(r(\theta)-\tilde{r}(\theta)\Big)\textsc{TS}(\theta, q_e)-  \Delta (\theta_h - \theta) ~\qquad~\textrm{(by definition of $\Delta$)}\\
          &= (\Delta - \epsilon)\textsc{TS}(\theta, q_e)-  \Delta (\theta_h - \theta) \\
          &= \Delta \big[\textsc{TS}(\theta, q_e) - (\theta_h - \theta)\big]-\epsilon \textsc{TS}(\theta, q_e)\\
          &= \Delta \textsc{TS}(\theta_h, q_e) -\epsilon \textsc{TS}(\theta, q_e)\\
          &\ge \Delta \textsc{TS}(\theta_h, q_e) -\epsilon \textsc{TS}(\underline{\theta}, q_e) \qquad~\textrm{(since $\textsc{TS}$ is decreasing in $\theta$)}\\
          &> 0
      \end{align*}
      where the last inequality follows from the fact that $\epsilon<\Delta\frac{\textsc{TS}(\theta_h, q_e)}{\textsc{TS}(\underline{\theta}, q_e)}$.
\end{proof}

\begin{lemma}\label{lem:Dr}
The set $D_{r}$ has a positive measure.    
\end{lemma}
\begin{proof}
Both $M$ and $\widetilde{M}$ are such that $r(\underline{\theta})=\tilde{r}(\underline{\theta})=1$ and $q(\underline{\theta})=\tilde{q}(\underline{\theta})=q_e$. Thus, the total surplus from both the mechanisms at $\underline{\theta}$ is the same. Because $\widetilde{M}$ dominates $M$, it must be that the rent paid by the monopolist of type $\underline{\theta}$ is lower in $\widetilde{M}$ than in $M$, that is,
\begin{align}\label{eq:rent}
    u(\underline{\theta}) = \int \limits_{\underline{\theta}}^{\bar{\theta}} r(z) dz \ge \int \limits_{\underline{\theta}}^{\bar{\theta}} \tilde{r}(z)dz = \tilde{u}(\underline{\theta}).
\end{align}
If the inequality is an equality, $r(\theta)=\tilde{r}(\theta)$ for almost all $\theta$, and by left-continuity of $M$ and $\widetilde{M}$, $r(\theta)=\tilde{r}(\theta)$ for all $\theta$. Because both $M$ and $\widetilde{M}$ are such that  $q(\underline{\theta})=\tilde{q}(\underline{\theta})=q_e$ for all $\theta$, this implies $M = \widetilde{M}$, which is a contradiction. Thus, inequality (\ref{eq:rent}) is strict, and $r(\theta) > \tilde{r}(\theta)$ for a positive measure of $\theta$. Consequently, $D_{r}$ has a positive measure.
\end{proof}

\end{document}